\documentclass[5p,times]{elsarticle}
\usepackage{ecrc}

\volume{00}

\firstpage{1}

\journalname{Procedia Computer Science}

\runauth{}

\jid{procs}

\jnltitlelogo{Procedia Computer Science}

\CopyrightLine{2011}{Published by Elsevier Ltd.}

\usepackage{amssymb}

 \biboptions{numbers,sort&compress}

\usepackage[figuresright]{rotating}

\usepackage{xspace}
\usepackage{multirow}
\usepackage{makecell} 
\usepackage{amsmath,amsfonts,amssymb}
\usepackage{amsthm}
\usepackage[ruled,vlined,linesnumbered]{algorithm2e} 
\usepackage{bm}
\usepackage{tabularx,booktabs}
\usepackage{mathtools}
\usepackage{url}
\usepackage{tikz}
\usepackage{threeparttable}

\def\CSP{C\&S'\xspace}
\def\CS{C\&S\xspace}
\def\DEF{CPX\xspace}
\def\CSES{C\&S\xspace}
\def\VMCP{VMCP\xspace}
\def\VMCPs{VMCPs\xspace}

\def\NVM{\#VM\xspace}

\def\TT{T\xspace}

\def\NC{\#CONS\xspace}
\def\NV{\#VARS\xspace}
\def\Solved{\#S\xspace}
\def\VM{VM\xspace}
\def\VMs{VMs\xspace}
\def\CPX{CPX\xspace}

\def\CASSIGN{c^{\rm alloc}_{i,k}}
\def\CRUN{c^{\rm run}_{k}}
\def\CMIG{c^{\rm mig}_{i,k}}

\def\UIR{u_{i,r}}
\def\SKR{s_{k,r}}
\def\DI{d_{i}}
\def\DINEW{d^{\rm new}_{i}}
\def\NIK{n_{i,k}}
\def\L{\text{\ell}}
\def\MK{m_{k}}
\def\ISET{\mathcal{I}}
\def\KSET{\mathcal{K}}
\def\RSET{\mathcal{R}}
\def\DSET{\mathcal{D}}
\def\XIK{x_{i,k}}
\def\YK{y_{k}}
\def\ZIK{z_{i,k}}
\def\XIKNEW{x^{\rm new}_{i,k}}

\newcommand{\MIP}{{MILP}\xspace}
\newcommand{\LP}{{LP}\xspace}
\newcommand{\solver}[1]{\textsc{#1}\xspace}
\newcommand{\cplex}{\solver{CPLEX}}
\newcommand{\gurobi}{\solver{Gurobi}}
\newcommand{\scip}{\solver{SCIP}}

\DeclareMathOperator{\conv}{conv}
\DeclareMathOperator*{\argmax}{arg\,max}

\usepackage{hyperref}

\newtheorem{proposition}{Proposition}

\begin{document}
\begin{sloppypar}  
\begin{frontmatter}



\title{A Cut-and-solve Algorithm for Virtual Machine Consolidation Problem}


\author[bupt]{Jiang-Yao Luo}
\ead{luoshui3000@bupt.edu.cn}
\author[amss]{Liang Chen}
\ead{chenliang@lsec.cc.ac.cn}
\author[bit]{Wei-Kun Chen}
\ead{chenweikun@bit.edu.cn}
\author[bupt]{Jian-Hua Yuan\corref{cor1}}
\ead{jianhuayuan@bupt.edu.cn}
\author[amss]{Yu-Hong Dai}
\ead{dyh@lsec.cc.ac.cn}

\cortext[cor1]{Corresponding author}

\address[bupt]{School of Science, Beijing University of Posts and Telecommunications, Beijing 100876, China}
\address[amss]{School of Mathematical Sciences, University of Chinese Academy of Sciences, Beijing, China}
\address[bit]{School of Mathematics and Statistics/Beijing Key Laboratory on MCAACI, Beijing Institute of Technology, Beijing 100081, China}

\begin{abstract}
	The virtual machine consolidation problem (\VMCP) attempts to determine which servers to be activated,
	how to allocate virtual machines (\VMs) to the activated servers, and how to migrate \VMs among servers such that the summation of activated, allocation, and migration costs is minimized subject to the resource constraints of the servers and other practical constraints.
	In this paper, we first propose a new mixed integer linear programming (\MIP) formulation for the \VMCP. 
	We show that compared with existing formulations, the proposed formulation is much more compact in terms of smaller numbers of variables or constraints, 
	which makes it suitable for solving large-scale problems.
	We then develop a cut-and-solve (\CS) algorithm, a tree search algorithm to efficiently solve the \VMCP to optimality.
	The proposed \CS algorithm is based on a novel relaxation of the \VMCP that provides a stronger lower bound than the natural continuous relaxation of the \VMCP, making a smaller search tree.
	By extensive computational experiments, we show that (i) the proposed formulation significantly outperforms existing formulations in terms of solution efficiency; and (ii) compared with standard \MIP solvers, the proposed \CS algorithm is much more efficient.

\end{abstract}


\end{frontmatter}




\section{Introduction}
\label{sect:introduction}
Nowadays, cloud computing provides a great flexibility and availability of computing resources to customers and has become more and more popular in many industries including the manufacturing industry \cite{xu2012cloud}, conference management systems \cite{ryan2011cloud}, and  E-commerce \cite{wang2013influences}.
The reason behind this success is that cloud service providers can provide customers with reliable, inexpensive, customized, and elastically priced computing resources without requiring customers to host them at a dedicated place.
In particular, various applications requested by different customers can be instantiated inside virtual machines (\VMs) and flexibly deployed to any server in any data center of the cloud \cite{Barham2003}. 
However, as \VMs change dynamically and run over a shared common cloud infrastructure, it is crucial to (re)allocate cloud resources (by migrating existing \VMs among different servers and/or mapping new \VMs into appropriate servers) to meet diverse application requirements while minimizing the operational costs of the service providers.
The above problem is called the \emph{virtual machine consolidation problem} (\VMCP) in the literature, which
determines the activation state of all servers and the (re)allocation of all \VMs in such a way that the predefined cost function
(server activation, \VM allocation, and migration costs) is minimized subject to resource constraints of the servers and other practical constraints.

The \VMCP is strongly NP-hard, as it includes the \emph{bin packing problem} (BPP) \cite{speitkamp2010mathematical}.
Therefore, there is no polynomial-time algorithm to solve the \VMCP to optimality unless P = NP.
As a result, most existing works investigated heuristic algorithms for solving the \VMCP.
In particular, the \emph{first-fit decreasing} and \emph{best-fit decreasing} algorithms, first investigated for BPPs, are used to solve \VMCPs; see \cite{beloglazov2012energy,speitkamp2010mathematical}.
Reference \cite{speitkamp2010mathematical} first formulated the \VMCP as a mixed integer linear programming (\MIP) problem and then proposed a linear programming (\LP) relaxation based heuristic algorithm.
This heuristic algorithm solves the LP relaxation of the \MIP problem, fixes the variables taking integral values, and tries to find a solution by solving the reduced \MIP problem. 
Reference \cite{goudarzi2012sla} proposed a heuristic algorithm based on the convex
optimization method and dynamic programming. 
Several metaheuristic algorithms were also developed to solve the \VMCP, 
including genetic algorithm \cite{wu2016energy,sharma2016multi,he2014developing}, simulated annealing \cite{marotta2015simulated}, 
colony optimization \cite{farahnakian2014using,jiang2017dataabc}, and evolution algorithm \cite{li2020energy}. 
However, the above heuristic algorithms cannot guarantee to find an optimal solution for the \VMCP.
Indeed, as found in \cite{mazumdar2017power,speitkamp2010mathematical,goudarzi2012sla},  the solutions found by the heuristic algorithms are 6\% to 49\%  far from the optimal solutions. 
Therefore, determining an optimal solution for the \VMCP is highly needed.

Usually, the \VMCP can be formulated as an \MIP problem,
which allows us to leverage state-of-the-art \MIP solvers such as \gurobi \cite{optimization2018gurobi} and \scip \cite{achterberg2009scip} to solve it to optimality.
In particular, in the formulations in \cite{speitkamp2010mathematical,marotta2015simulated,beloglazov2012optimal,ferreto2011server,beloglazov2012energy,laili2018iterative,wolke2015more,wu2016energy,mann2015allocation,helali2021survey}, the authors used binary variables to denote whether a given \VM is mapped into a given server, and presented the constraints and the objective function based on these binary variables.
One weakness of these formulations is that the problem size grows linearly with the number of \VMs.
When the number of \VMs is large, these formulations are difficult to be solved by standard \MIP solvers.
In practice, the requested loads of many \VMs are identical, meaning that the number of \VM types is relatively small even when the number of \VMs is huge  \cite{AlibabaCluster}.
Reference \cite{mazumdar2017power}  took this observation into account and proposed a formulation using a family of integer variables, which represents the number of \VMs of a given type on a given server. 
As a result, the problem size of this formulation grows linearly with the number of \VM types but not the number of \VMs.
However, in order to model the migration process, the authors used a family of 3-index integer variables indicating the number of \VMs of a given type migrated from one server to another server.
Due to this family of variables, the problem size grows quadratically with the number of servers, making it unrealistic to solve this formulation by standard \MIP solvers within a reasonable time limit, especially when the number of servers is large. 

To summarize, the existing formulations for the \VMCP suffer from a large problem size when the number of \VMs is large or the number of servers is large.
This fact makes it difficult to (i) employ a standard \MIP solver to solve the \VMCP within a reasonable time limit, and (ii) develop an efficient customized exact algorithm for the \VMCP (as such algorithms are usually based on a formulation with a small problem size).  
The motivation of this work is to fill this research gap. In particular,
\begin{itemize}
	\item[1)] We present new \MIP formulations for \VMCPs, which minimizes the summation of server activation, \VM allocation, and migration costs subject to resource constraints and other practical constraints.
	The proposed new formulation is much more compact than
	the existing formulations in \cite{speitkamp2010mathematical} 
	and \cite{mazumdar2017power} in terms of the smaller number of variables or constraints. 
	\item[2)] We develop a cut-and-solve algorithm (called \CSES) 
	to solve \VMCPs to optimality based on the new formulations.
	The proposed \CS algorithm is based on a novel relaxation of the \VMCP that provides a stronger lower bound than the natural continuous relaxation of the \VMCP, making a smaller search tree.
\end{itemize}
Extensive computational results demonstrate that (i) the proposed formulation significantly outperforms existing formulations in terms of solution efficiency; and (ii) compared with standard \MIP solvers, the proposed \CS algorithm is much more efficient.


The paper is organized as follows. 
Section \ref{sect:problem_formualtion} introduces the novel \MIP formulations for \VMCP and compares them with the formulations in \cite{speitkamp2010mathematical} and \cite{mazumdar2017power}.
Section \ref{sect:solution_methodology} describes the proposed \CSES algorithm to solve \VMCPs.
Section \ref{sect:numerical_results} presents the computational results.
Finally, Section \ref{sect:conclusion_remarks} draws some concluding remarks. 
\section{Virtual machine consolidation problems}
\label{sect:problem_formualtion}
The \VMCP attempts to determine which servers to be activated,
how to allocate \VMs to the activated servers, and how to migrate \VMs among servers such that the sum of server activation, \VM allocation, and migration costs is minimized subject to the resource constraints of the servers and other practical constraints.
In this section, we first present an \MIP formulation for a basic version of the \VMCP (in which only the resource constraints of the servers are considered). 
Then, we present a variant of the \VMCP that considers other practical constraints. 
Finally, we show the advantage of the proposed formulation by comparing it with those in Speitkamp and Bichler \cite{speitkamp2010mathematical} 
and Mazumdar and Pranzo \cite{mazumdar2017power}.
\subsection{The basic virtual machine consolidation problem}
\label{sect:vm_consolidation_problem}
\begin{table}[htbp]
	\centering
	\setlength{\tabcolsep}{3pt}
	\renewcommand{\arraystretch}{1.1}
	\caption{Summary of parameters and variables.}
	\label{sect:vmcp}
	\begin{tabular}{|c|l|}
		\hline
		\multicolumn{2}{|c|}{Parameters} \\
		\hline
		$\UIR$& \makecell[lt]{units of resource $r$ requested by a \VM of type $i$}\\
		$\SKR$& \makecell[lt]{units of resource $r$ provided by server $s$}\\
		$\CASSIGN$& \makecell[lt]{cost of allocating a \VM of type $i$ to server $k$}\\
		$\CRUN$& \makecell[lt]{activation cost of server $k$}\\
		$\CMIG$& \makecell[lt]{cost of migrating a \VM of type $i$ to server $k$}\\
		$\DI$& \makecell[lt]{number of \VMs of type $i$ that need to be allocated}\\
		$\NIK$& \makecell[lt]{number of \VMs of type $i$ that (currently) be allocated\\ to server $k$}\\
		$\L$& \makecell[lt]{maximum number of allowed migrations}\\
		$\MK$& \makecell[lt]{maximum number of \VMs allocated to server $k$}\\
		$\DINEW$& \makecell[lt]{number of new incoming \VMs of type $i$ that need to\\ be allocated}\\
		\hline
		\multicolumn{2}{|c|}{Variables} \\
		\hline
		$\XIK$& \makecell[lt]{integer variable representing the number of \VMs of \\ type $i$ allocated to server $k$}\\
		$\YK$& \makecell[lt]{binary variable indicating whether or not server $k$ is \\activated}\\
		$\ZIK$& \makecell[lt]{integer variable representing the number of \VMs of \\ type $i$ migrated to server $k$}\\
		$\XIKNEW$& \makecell[lt]{integer variable representing the number of new \\ incoming \VMs of  type $i$ allocated to server $k$}\\
		\hline
	\end{tabular}
\end{table}

Let $\KSET$, $\ISET$, and $\RSET$ denote the set of the servers,
the set of types of \VMs that need to be allocated to the servers,
and the set of resources (e.g., CPU, RAM, and Bandwidth \cite{speitkamp2010mathematical}) of the servers, respectively.
Each server $k$ can provide $\SKR$ units of resource $r$ and 
each \VM of type $i$ requests $\UIR$ units of resource $r$.
Before the \VM consolidation, there are $\NIK$ \VMs of type $i$ that are currently allocated to server $k$.
For notations purpose, we denote $\DI = \sum_{k \in \KSET} \NIK$ for all $i \in \ISET$.
We introduce integer variable $\XIK$ to represent the number of \VMs of type $i$ allocated to server $k$ (after the \VM consolidation), 
binary variable $\YK$ to indicate whether or not server $k$ is activated,
and integer variable $\ZIK$ to represent the number of \VMs of type $i$ migrated to server $k$.
Then the mathematical formulation of the \VMCP can be written as:
\begin{subequations}
\label{eq:mathforms}
\setlength{\abovedisplayskip}{5pt}
\setlength{\belowdisplayskip}{5pt}
\begin{align} 
\!\!\!\! \mathop{\mathrm{min}} \ 
&\sum_{k \in \KSET} \CRUN \YK 
+\sum_{i \in \ISET}\sum_{k \in \KSET} \CASSIGN \NIK  + \sum\limits_{i \in \ISET}\sum\limits_{k \in \KSET} \CMIG \ZIK  \label{eq:obj}\\
\!\!\!\! \text{s.t.} ~~& \sum\limits_{i \in \ISET} \UIR \XIK \leq \SKR \YK,~\forall~k \in \KSET,~  
\forall~r \in \RSET,  \label{eq:resource}   \\
& 	\sum\limits_{k \in \KSET} \XIK = \DI, ~ \forall~i \in \ISET,   \label{eq:allocation_general}  \\
&  (\XIK -  \NIK)^{+} = \ZIK , ~ \forall~i \in \ISET,~\forall~k \in \KSET,  \label{eq:migration} \\
& \XIK,~\ZIK \in \mathbb{Z_+},~ \XIK \leq v_{i,k},~\forall~i \in \ISET,~\forall~k \in \KSET, \label{eq:varX}\\
&  \YK \in \{0,1\},~\forall~k \in \KSET. \label{eq:varY}
\end{align}
\end{subequations}

Constraint \eqref{eq:resource} ensures that the total workload of \VMs allocated to each server does not exceed any of its resource capacity.
Constraint \eqref{eq:allocation_general} enforces that all \VMs of each type have to be allocated to servers. 
Constraint \eqref{eq:migration} relates variables $\XIK$ and $\ZIK$. 
More specifically, it enforces that if the number of \VMs of type $i$ allocated to server $k$ after the \VM consolidation, $\XIK$, is larger than that before the \VM consolidation, $\NIK$,
then the number of \VMs of type $i$ migrated to server $k$, $\ZIK$, must be equal to $\XIK-\NIK$; otherwise, it is equal to zero.
Finally, constraints \eqref{eq:varX} and \eqref{eq:varY} enforce $\XIK$, $\ZIK$, and $\YK$ to be integer/binary variables and trivial upper bounds $\{v_{i,k}\}$ for variables $\left\{\XIK\right\}$ where 
\begin{equation*}
	v_{i,k} = \min\left\{\sum_{k \in \KSET}n_{i,k}, ~\min_{r \in \RSET} \left \{ \left\lfloor  \frac{s_{k,r}}{u_{i,r}} \right\rfloor \right  \}\right\}, ~\forall~i \in \ISET, ~\forall~k \in \KSET.
\end{equation*}

The objective function \eqref{eq:obj} to be minimized is 
the sum of the activation cost of servers, 
the cost of allocating all \VMs to servers, 
and the cost of migrating \VMs among servers. 
Here $\CRUN \geq 0$, $\CASSIGN\geq 0$, and $\CMIG\geq 0$ denote the activation cost of server $k$, 
the cost of allocating a \VM of type $i$ to server $k$, and the cost of migrating a \VM of type $i$ to server $k$, respectively.
In practice, $\CRUN$ and $\CASSIGN$ reflect the power consumption of activating server $k$ and allocating a \VM of type $i$ to server $k$, respectively \cite{mazumdar2017power}.
As demonstrated in \cite{mazumdar2017power}, the \VM migration process creates non-negligible energy overhead on the source and destination servers (see also \cite{dargie2014estimation,rybina2013investigation}).
Therefore, we follow Mazumdar and Pranzo \cite{mazumdar2017power} to consider the energy cost as the migration cost and assume $ \CMIG=\CASSIGN$ for all $i \in \ISET$ and $k\in \KSET$.
Notice that for a \VM of type $i \in \ISET$ that is previously hosted at server $k \in \KSET$, if it is still run on server $k$ after the \VM consolidation, it will only incur the allocation cost at the source server; and if it is migrated to a destination server $k'$, it will incur the allocation costs at the source and destination servers $k$ and $k'$ (as $ \CMIG=\CASSIGN$ for all $i \in \ISET$ and $k\in \KSET$).

Problem \eqref{eq:mathforms} is an \MIP problem since the nonlinear constraint \eqref{eq:migration} can be equivalently linearized. 
Indeed, by $\CMIG \geq 0$, constraint \eqref{eq:migration} can be equivalently presented as the following \emph{linear} constraint
\begin{equation}
	\tag{1d'}
	\XIK - \NIK \leq  \ZIK , ~ \forall~i \in \ISET,~\forall~k \in \KSET.  \label{eq:migration1} \\
\end{equation}
Note that the linearity of all variables in problem \eqref{eq:mathforms} is vital, which enables to leverage the efficient \MIP solver such as \cplex \cite{CPLEX} to solve the problem to global optimality.

\subsection{Extensions of the virtual machine consolidation problem}
\label{sect:extensions_of_vmcp}
The \VMCP attempts to (re)allocate \VMs to servers subject to resource constraints.  
In practice, however, a \VM manager should also deal with other practical requirements.
In this subsection, we introduce four side constraints derived from practical applications in the literature
including new incoming \VMs \cite{mazumdar2017power}, 
restriction on the maximum number of \VM migrations \cite{mazumdar2017power,speitkamp2010mathematical,bichler2006capacity}, 
restriction on the maximum number of \VMs on servers \cite{mazumdar2017power,anselmi2008service}, 
and restriction on allocating a \VM type to a server \cite{anselmi2008service,dhyani2010constraint}.
All these constraints can be incorporated into the \VMCP.\\[5pt]
{\noindent $\bullet$ New incoming \VMs}\\[5pt]   
\indent The cloud data center needs to embed new incoming \VMs into the servers \cite{mazumdar2017power}. 
We denote the number of new incoming \VMs of type $i$, $i \in \ISET$, as $\DINEW$.
To deal with new incoming \VMs, we introduce integer variable $\XIKNEW$ to denote the number of new incoming \VMs of type $i$ allocated to server $k$.
To ensure all new incoming \VMs are allocated to servers, we need constraints
\begin{align}
	\sum\limits_{k \in \KSET} \XIKNEW = \DINEW,~\forall~i \in \ISET. \label{eq:new_imcoming_vms}
\end{align}
In addition, the term $\sum_{i \in \ISET}\sum_{k \in \KSET} \CASSIGN \XIKNEW$ must be included into the objective function of problem \eqref{eq:mathforms} to reflect 
the cost of allocating all new \VMs to the servers.
Moreover, as the embedding new incoming \VMs into a server also leads to resource consumption, the capacity constraint \eqref{eq:resource} must be changed into
\begin{align}
	&\!\!\!\!\! \sum\limits_{i \in \ISET} \UIR \XIK + \sum\limits_{i \in \ISET} \UIR\XIKNEW \leq \SKR \YK,~\forall~k \in \KSET,~  
	\forall~r\in \RSET. \label{eq:transformed_resource}
\end{align}
{\noindent $\bullet $ Maximum number of \VM migrations}\\[5pt]
\indent Due to limit administrative costs, the cloud data center manager requires that the number of migrations of \VMs cannot exceed a predefined number $\L$ \cite{mazumdar2017power,bichler2006capacity,speitkamp2010mathematical}, 
which can be enforced by 
\begin{align}
	\sum\limits_{i \in \ISET}\sum\limits_{k \in \KSET} \ZIK \leq \L, \label{eq:limit_number_vm_migrations} 
\end{align}
{\noindent $\bullet$ Maximum number of \VMs on servers}\\[5pt] 
\indent In practice, the cloud data center manager may spend a lot of time in the event of a server failure if too many \VMs are allocated to the server \cite{mazumdar2017power,anselmi2008service}.
Consequently, it is reasonable to impose a threshold $\MK$ on the maximum number of  \VMs that are allocated to server $k \in \KSET$:
\begin{align}
	\sum_{i \in \ISET} \XIK \leq \MK ,~\forall~k \in \KSET. \label{eq:max_number_vms}
\end{align}
If new incoming \VMs are also required to be embedded in the servers, then constraint \eqref{eq:max_number_vms} should be rewritten as 
\begin{equation}
	\sum_{i \in \ISET} \XIK + \sum_{i \in \ISET} \XIKNEW \leq \MK ,~\forall~k \in \KSET. \label{eq:max_number_vms1}
\end{equation}
{\noindent $\bullet$ Allocation restriction constraints}\\[5pt]  
\indent A subset of servers may exhibit some properties, 
such as kernel version, clock speed, and the presence of an external IP address.
It is impossible to allocate \VMs with specific attribute requirements to a server that does not provide these attributes \cite{anselmi2008service,dhyani2010constraint}.
This can be enforced by constraint 
\begin{align}
		\label{eq:attribute}
	\XIK = 0,~\forall~i \in \ISET,~\forall~k \in \mathcal{K}(i),
\end{align}
where $\mathcal{K}(i) \subseteq \KSET$ denotes the set of servers that cannot process \VMs of type $i$. 
Similarly, if new incoming \VMs are also required to be embedded in the servers, then constraint 
\begin{equation}
	\label{eq:attribute1}
	\XIKNEW = 0,~\forall~i \in \ISET,~\forall~k \in \mathcal{K}(i),
\end{equation}
needs to be included in problem \eqref{eq:mathforms}.

\subsection{Comparison with the formulations in \cite{speitkamp2010mathematical} and \cite{mazumdar2017power}}
\label{sect:comparison_with_two_other_formulations}
To formulate the \VMCP or its extensions, other \MIP formulations in the literature can be used.
In this subsection, we briefly review the two \MIP formulations in \cite{speitkamp2010mathematical} and \cite{mazumdar2017power} 
and show the advantages of the proposed formulation over these two existing formulations.
For easy presentation and fair comparison, we change the objective functions and constraints of the problems in \cite{speitkamp2010mathematical} and \cite{mazumdar2017power} to 
be the same as those of the basic \VMCP in Section \ref{sect:vm_consolidation_problem} \footnote{
	We remark that (i) the problem in \cite{speitkamp2010mathematical} attempts to allocate new incoming \VMs to the servers such that the activation cost is minimized subject to the resource constraint; (ii)
	and the problem in \cite{mazumdar2017power} is an extension of the basic \VMCP \eqref{eq:mathforms} in which the new incoming \VMs (constraints \eqref{eq:new_imcoming_vms}-\eqref{eq:transformed_resource}) and the limitation on the number of \VM migrations (constraint \eqref{eq:limit_number_vm_migrations}) are considered.}.

First, we compare the proposed formulation \eqref{eq:mathforms} with the formulation in \cite{speitkamp2010mathematical}.
Different from our proposed formulation where a 2-index integer variable $\XIK$ is used to represent the number of \VMs of type $i$ allocated to server $k$, in the formulation of  \cite{speitkamp2010mathematical}, a 3-index binary variable $x_{i,v,k}$ is used to represent whether or not \VM $v$ of type $i$ is allocated to server $k$ after \VM consolidation.
Similarly, $n_{i,v,k}$ and $z_{i,v,k}$ are used to represent whether \VM $v$ of type $i$ is allocated to server $k$ before \VM consolidation and 
whether or not \VM $v$ of type $i$ is migrated to server $k$, respectively.
The formulation in  \cite{speitkamp2010mathematical} can be written as
\begin{subequations}
	\label{eq:mathforms_sb}
	\setlength{\abovedisplayskip}{5pt}
	\setlength{\belowdisplayskip}{5pt}
	\begin{align} 
	\mathop{\mathrm{min}} \ 
	&\sum_{k \in \KSET} \CRUN \YK 
	+\sum_{i \in \ISET}\sum_{v \in \DSET(i)}\sum_{k \in \KSET} \CASSIGN n_{i,v,k}   + \sum\limits_{i \in \ISET}\sum_{v \in \DSET(i)}\sum\limits_{k \in \KSET} \CMIG z_{i,v,k}  \label{eq:obj_sb}\\
	\text{s.t.} ~~& \sum\limits_{i \in \ISET} \sum_{v \in \DSET(i)} \UIR x_{i,v,k} \leq \SKR \YK,~\forall~k \in \KSET,~  
	\forall~r \in \RSET,  \label{eq:resource_sb}   \\
	& \sum\limits_{k \in \KSET} x_{i,v,k} = 1, ~ \forall~i \in \ISET,~\forall~v \in \DSET(i),   \label{eq:allocation_sb}  \\
	&  x_{i,v,k} - n_{i,v,k} \leq z_{i,v,k} , ~ \forall~i \in \ISET,~\forall~v \in \DSET(i),~\forall~k \in \KSET,  \label{eq:migration_sb} \\
	& x_{i,v,k}, z_{i,v,k} \in \{0,1\},~\forall~i \in \ISET,~\forall~v \in \DSET(i),~\forall~k \in \KSET, \label{eq:varX_sb} \\
	& \YK \in \{0,1\},~\forall~k \in \KSET, \label{eq:varY_sb} 
	\end{align}
\end{subequations} 
where $\DSET(i)=\left\{1, \ldots, \DI\right\}$. 
Though formulations \eqref{eq:mathforms} and \eqref{eq:mathforms_sb} are equivalent in terms of returning the same optimal solution, 
the problem size of the proposed formulation \eqref{eq:mathforms}  is much smaller than that of \eqref{eq:mathforms_sb}. 
Indeed, the number of variables and constraints in \eqref{eq:mathforms} are $\mathcal{O}(|\KSET||\ISET|)$ and $\mathcal{O}(|\KSET|(|\ISET|+|\RSET|))$, respectively, while those in
\eqref{eq:mathforms_sb} are $\mathcal{O}(|\KSET|\sum_{i \in \ISET}\DI)$ and $\mathcal{O}(|\KSET|(\sum_{i \in \ISET}\DI+|\RSET|))$, respectively. 
In practice, the requested loads of many \VMs such as CPU cores and normalized memory, are identical \cite{AlibabaCluster}, implying that $\left|\mathcal{I}\right| \ll \sum_{i \in \ISET}\DI$.

Next, we compare the proposed formulation \eqref{eq:mathforms} with the formulation in \cite{mazumdar2017power}. 
Different from the proposed formulation \eqref{eq:mathforms} where a 2-index integer variable $\ZIK$ is used to represent the number of \VMs of type $i$ migrated to server $k$, in the formulation in \cite{mazumdar2017power}, a 3-index integer variable $z_{i,j,k}$ is used to indicate the number of \VMs of type $i$ migrated from server $j$ to server $k$.
The mathematical formulation in \cite{mazumdar2017power} can be presented as 
\begin{subequations}
\setlength{\abovedisplayskip}{5pt}
\setlength{\belowdisplayskip}{5pt}
\label{eq:mathforms_int}
\begin{align}
\min
& \sum_{k \in \KSET} \CRUN \YK+  \sum_{i \in \ISET}\sum_{k \in \KSET} \CASSIGN \NIK + \sum_{i \in \ISET}\sum_{j \in \KSET}\sum_{k \in \KSET} \CMIG z_{i,j,k}  \label{eq:obj_int}  \\
 \text{s.t.} ~& \NIK + \sum_{j \in \KSET} (z_{i,j,k}-z_{i,k,j}) \geq 0,~\forall~ i\in \ISET,~ \forall ~k \in \KSET,\\
& \sum\limits_{i \in \ISET} \UIR \left (\NIK + \sum_{j \in \KSET}(z_{i,j,k} - z_{i,k,j})\right) \leq \SKR \YK, \nonumber \\
&\qquad\qquad\qquad\qquad\qquad\qquad~ \forall ~k \in \KSET,~\forall~ r\in \RSET,   \label{eq:resource_int}  \\
& z_{i,j,k}\in \{0, \ldots, \NIK\},~\forall~ i \in \ISET,~\forall~ j \in \KSET,~\forall~ k \in \KSET, \label{eq:varXZ_int} \\
&\YK \in \{0,1\},~\forall~ k \in \KSET. \label{eq:varY_int}
\end{align}
\end{subequations}
Though both numbers of constraints in \eqref{eq:mathforms_int} and \eqref{eq:mathforms} are   $\mathcal{O}(|\KSET|(|\ISET|+|\RSET|))$,
the number of variables in \eqref{eq:mathforms_int}, however, is much larger than that in \eqref{eq:mathforms} ($\mathcal{O}(|\KSET|^2|\ISET|)$ versus $\mathcal{O}(|\KSET||\ISET|)$).

Based on the above discussion, we can conclude that the proposed formulation \eqref{eq:mathforms} for the \VMCP is much more compact than the two existing formulations \eqref{eq:mathforms_sb} and \eqref{eq:mathforms_int} in \cite{speitkamp2010mathematical} and \cite{mazumdar2017power}.
Therefore, formulation \eqref{eq:mathforms} can be much easier to solve than formulations  \eqref{eq:mathforms_sb} and \eqref{eq:mathforms_int} by standard \MIP solvers (e.g., \gurobi and \cplex), as demonstrated in the Section \ref{sect:comparison_with_other_two_mathematical_formulations}.
In addition, as shown in \cite[Theorem 1]{speitkamp2010mathematical}, the basic \VMCP is strongly NP-hard even for the case $\CMIG = 0$ for all $i \in \ISET$ and $k \in \KSET$, meaning that customized (exact or heuristic) algorithms for efficiently solving the \VMCP is needed in practice, especially when the problem's dimension is large. 
We remark that compact formulation \eqref{eq:mathforms} is an important step towards developing an efficient customized algorithm for solving the \VMCP (e.g., it will lead to a compact \LP relaxation, which is the basis of many efficient customized \LP relaxation based algorithms).
In the next section, we shall develop an efficient customized exact algorithm based on formulation \eqref{eq:mathforms} for solving the \VMCP.

\section{The cut-and-solve algorithm}
\label{sect:solution_methodology} 
In this section, we shall develop a cut-and-solve (\CS) algorithm to efficiently obtain an optimal solution of the \VMCP.
Specifically, in Section \ref{sect:cut-and-solve}, we demonstrate how to apply the \CS procedure \cite{climer2006cut}, a tree search procedure, to solve the \VMCP.
Then, in Section \ref{sect:exact_separation_approach}, we propose a new relaxation for the \VMCP and develop a cutting plane approach to solve this new relaxation.
The newly proposed relaxation provides a stronger lower bound than the natural continuous relaxation of the \VMCP, which effectively reduce the search tree size, making a more efficient \CS algorithm.
For simplicity of presentation, we only present the algorithm for the basic \VMCP in Section \ref{sect:vm_consolidation_problem} as the proposed algorithm can easily be adapted to solve the extension of the \VMCP in Section \ref{sect:extensions_of_vmcp}.

\subsection{Cut-and-solve}
\label{sect:cut-and-solve}
The \CS procedure, first proposed by Climer and Zhang \cite{climer2006cut}, has been applied to solve well-known structured mixed binary programming problems, e.g., the traveling salesman problem \cite{climer2006cut}, the facility location problem \cite{yang2012cut,yang2019effective,gadegaard2018improved}, and the multicommodity uncapacitated fixed-charge network design problem \cite{Zetina2019}.
For these problems, \CS has been demonstrated to be much more efficient than generic \MIP solvers. 
In this subsection, we shall demonstrate how to apply the \CS procedure to solve \VMCPs. 
\begin{figure}[htbp]
	\centering
	\includegraphics[width=\linewidth]{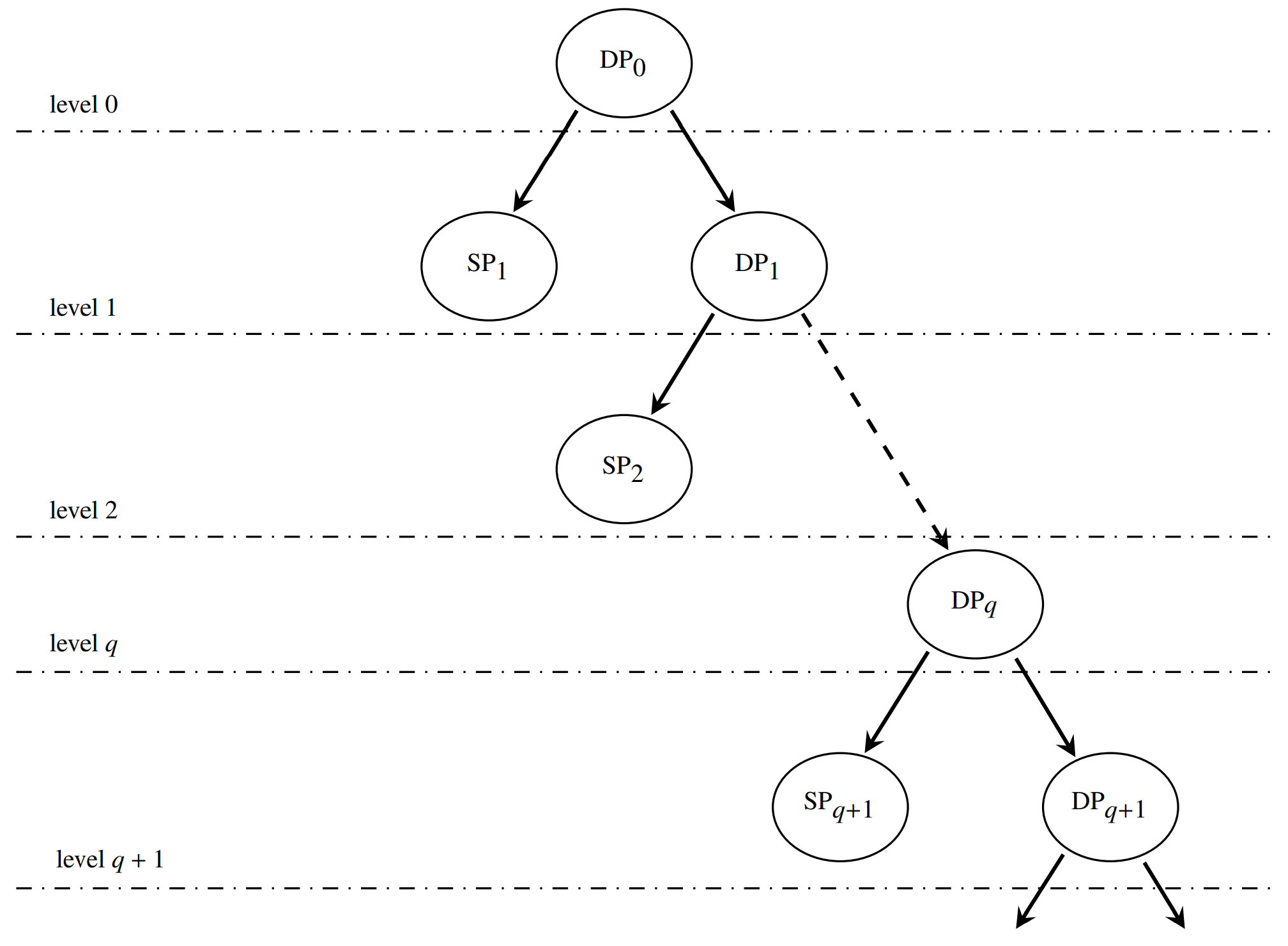}
	\caption{\CS search tree}
	\label{fig:cut-and-solve_search_tree}
\end{figure}

The \CS procedure is essentially a branch-and-bound algorithm in which a search tree (see Fig. \ref{fig:cut-and-solve_search_tree}) will be constructed during the search.
More specifically, the \CS denotes the original problem \eqref{eq:mathforms} as $\text{DP}_0$ and the best objective value of all feasible solutions found so far as $\text{UB}_{\min}$ (the corresponding feasible solution is denoted as $(x_{\min}, y_{\min}, z_{\min})$).
At the $q$-th ($q\in \mathbb{Z}_+$) level of the \CS search tree (see Fig. \ref{fig:cut-and-solve_search_tree}),
we first solve the \LP relaxation of $\text{DP}_q$, denoted its solution by $(x^*, y^*,z^*)$ and the corresponding objective value by $\text{LB}_{q}$. 
 \begin{itemize}
	\item [i)] If $(x^*,y^*,z^*)$  is an integer vector, then $(x^*,y^*,z^*)$ is an optimal solution of problem $\text{DP}_{q}$ and hence the searching procedure can be terminated;
	\item [ii)] If $\text{LB}_{q} \geq \text{UB}_{\min}$, then $(x_{\min}, y_{\min}, z_{\min})$ must be an optimal solution of problem \eqref{eq:mathforms} and hence the searching procedure can also be terminated.
\end{itemize}
If neither (i) nor (ii) is satisfied, we then decompose the $\text{DP}_{q}$ into two subproblems $\text{SP}_{q+1}$ and $\text{DP}_{q+1}$, which are defined as $\text{DP}_{q}$ with the so-called \emph{piercing cuts} \cite{climer2006cut}
\begin{equation}
	\sum\limits_{k \in \mathcal{S}} y_k  \leq \varphi, \label{eq:fix}
\end{equation}
and
\begin{equation}
	\sum\limits_{k \in \mathcal{S}} y_k  \geq \varphi + 1, \label{eq:piercing_cut}
\end{equation}
respectively. 
Here $\varphi $ is a positive integer and $\mathcal{S}$ is a subset of $\KSET$.
A common selection used in \cite{climer2006cut,yang2012cut,yang2019effective,gadegaard2018improved} is to set $\varphi = 0$, making the right subproblem $\text{DP}_{q+1}$ and the left subproblem $\text{SP}_{q+1}$ being a \emph{dense problem} and a \emph{sparse problem} (in terms of small solution space). 
Indeed, for $\text{SP}_{q+1}$, such a selection forces (i) $\YK = 0$ for all $k \in \mathcal{S}$ and (ii) $\XIK = 0$ and $\ZIK = 0$ for all $i \in \ISET$ and $k \in \mathcal{S}$ (implied by constraints \eqref{eq:resource} and \eqref{eq:migration1}).
Due to the small solution space, sparse problem $\text{SP}_{q+1}$ can be solved by standard \MIP solvers within a reasonable time limit, especially when $|\mathcal{S}|$ is large.  
Moreover, as far as the sparse problem $\text{SP}_{q+1}$ has a feasible solution, it can provide an upper bound $\text{UB}_{q+1}$ for problem \eqref{eq:mathforms}.
If $\text{UB}_{q+1} < \text{UB}_{\min}$, $\text{UB}_{\min} \coloneqq \text{UB}_{q+1}$ will be updated.  
The above procedure is repeated until case i) or ii) is satisfied.
The details are summarized in Algorithm \ref{alg:cses_algorithm}.


We now discuss the selection of $\mathcal{S}$ in \eqref{eq:fix} and \eqref{eq:piercing_cut}, which defines subproblems $\text{SP}_{q+1}$ and $\text{DP}_{q+1}$. 
A straightforward strategy is to choose the set $\mathcal{S}$ as
\begin{equation}
	\mathcal{S}_1 := \left\{ k \in \KSET\, : \, y^*_k = 0  \right\},
\end{equation}
where $y^*$, as stated, is the optimal solution of $\text{DP}_{q}$'s relaxation.
The rationale behind strategy lies in the fact that an optimal solution to $\text{DP}_{q}$ usually has a number of components that are identical to those of an optimal solution to $\text{DP}_{q}$'s relaxation.
Consequently, it is more likely to find a good feasible solution (in terms of small objective value) by solving the sparse problem $\text{SP}_{q+1}$. 
However, our preliminary experiments showed that due to the (general) degeneracy of the \LP relaxation of $\text{DP}_{q}$, such a simple strategy cannot improve the lower bound (returned by solving the \LP relaxation of $\text{DP}_{q+1}$) fast enough, leading to a large search tree. 
For this reason, we use a more sophisticated strategy, suggested by Climer and Zhang \cite{climer2006cut}, to determine  $\mathcal{S}$, which is detailed as follows.
From the basic \LP theory, the reduced cost $r^*_k$ of a variable $\YK$ is a lower bound on the increase of the objective value if the value of this variable is changed by one unit. 
These reduced costs can be obtained by solving the \LP relaxation of $\text{DP}_{q+1}$.
Moreover, 
\begin{itemize}
	\item [1)] If $y_k^*=0$, then $r_k^* \geq 0$; 
	\item [2)] If $y_k^* = 1$, then $r_k^* \leq 0$;   
	\item [3)] If $0 < y_k^* < 1$, then $r_k^*=0$.
\end{itemize}
For more details, we refer to \cite[Chapter 5]{Dantzig2003}. 
Using the reduced costs, we set $\mathcal{S}$ as 
\begin{equation}
	\mathcal{S}_2:= \left\{ k \in \KSET\, : \, r^*_k \geq \epsilon   \right\},
\end{equation}
where $\epsilon> 0$ that controls the problem size of $\text{SP}_{q+1}$ and the size of \CS's search tree.
Indeed, it follows that $\mathcal{S}_2 \subseteq \mathcal{S}_1$, leading to a relatively large sparse problem $\text{SP}_{q+1}$, as compared to that defined by $\mathcal{S}_1$.
The larger the $\epsilon$ is, the larger the solution space of $\text{SP}_{q+1}$ is.
However, this also enables to obtain a much large better lower bound returned by solving $\text{DP}_{q+1}$'s relaxation (indeed, the difference of the objective values $\text{DP}_{q}$'s and $\text{DP}_{q+1}$'s relaxations is at least $\min\left\{ r_k\, : \, k \in \mathcal{S}_2 \right\} \geq \epsilon > 0$), yielding a smaller \CS search tree.
In our implementation, we set $\epsilon = 10^{-4}$.
\begin{algorithm}[t]
	\caption{\CS algorithm} 
	\label{alg:cses_algorithm}
	Initialize $\text{UB}_{\min} \coloneqq +\infty$ and $q \coloneqq 0$;\\
	\While{true}
	{
		Solve the \LP relaxation of problem $\text{DP}_q$ with the solution $(x^*, y^*, z^*)$ and the objective value $\text{LB}_{q}$;\\
		\If{$\text{\rm LB}_{q} \geq \text{\rm UB}_{\min}$}
		{
			Stop and return the optimal solution $(x_{\min}, y_{\min}, z_{\min})$;
		}
		\If{$(x^*, y^*, z^*)$ is an integer vector}
		{
			Stop and return the optimal solution $({x}^*, {y}^*, {z}^*)$;
		}
		Use the piercing cuts \eqref{eq:fix} and \eqref{eq:piercing_cut} to decompose $\text{DP}_q$ into two subproblems $\text{SP}_{q+1}$ and $\text{DP}_{q+1}$;\\
		Solve $\text{SP}_{q+1}$ to optimality by an \MIP solver and denote its solution and objective value by $(x',y',z')$ and $\text{UB}_{q+1}$;\\
		\If{$\text{\rm UB}_{q+1} < \text{\rm UB}_{\min}$}
		{
			Update $(x_{\min}, y_{\min}, z_{\min})\coloneqq(x',y',z')$ and $\text{UB}_{\min} \coloneqq \text{UB}_{q+1}$;\\
		}
		Set $q \coloneqq  q+1$;
	}
\end{algorithm}
\subsection{An improved relaxation problem and the cutting plane approach}
\label{sect:exact_separation_approach}
The \LP relaxation of problem \eqref{eq:mathforms}, obtained by ignoring the integrality requirement on all decision variables, is as follows:
\begin{subequations}
	\label{eq:mathforms_lp}
	\begin{align} 
		\mathop{\mathrm{min}} \ 
		&\sum_{i \in \ISET}\sum_{k \in \KSET} \CASSIGN \NIK  +\sum_{k \in \KSET} \CRUN \YK 
		+ \sum\limits_{i \in \ISET}\sum\limits_{k \in \KSET} \CMIG \ZIK  
		\\
		\text{s.t.} ~~& 
		\eqref{eq:resource},\eqref{eq:allocation_general}, \eqref{eq:migration1}, & \label{eq:vmcp_lp}\\
		& \XIK,~\ZIK \in \mathbb{R_+},~ \XIK \leq v_{i,k},~\forall~i \in \ISET,~\forall~k \in \KSET, \label{eq:varX_lp}\\
		&  \YK\in [0,1],~\forall~k \in \KSET. \label{eq:varY_lp}
	\end{align} 
\end{subequations}
However, the feasible region of \LP relaxation \eqref{eq:mathforms_lp} is actually enlarged compared with that of the original problem \eqref{eq:mathforms}.
As a result, this relaxation usually provides a weak lower bound, leading to a large \CS search tree; see Section \ref{sect:efficiency_of_cses} further ahead.
To overcome this weakness, we present a new relaxation that has a much more compact feasible region and hence can provide a much stronger lower bound, as compared with relaxation \eqref{eq:mathforms_lp}.
Then we provide a cutting plane approach to solve this newly proposed relaxation. 

\subsubsection{An improved relaxation problem}
\label{sect:motivation}


To proceed, we observe that in problem \eqref{eq:mathforms}, it is required that
\begin{align}
	 (x_{\cdot,k},\YK) &\in \mathcal{X}(r,k):=\bigg\{ (x_{\cdot, k},\YK)\in\mathbb{Z}_{+}^{|\ISET|}\times\{0,1\}:  \nonumber\\
	  & \sum\limits_{i \in \ISET} \UIR \XIK \leq \SKR \YK,~\XIK \leq v_{i,k},~ \forall~i \in \ISET ~\bigg \},\label{integerX}
\end{align}
for all $k \in \KSET$ and $r \in \RSET$.
However, in problem \eqref{eq:mathforms_lp}, such a requirement is relaxed to
\begin{equation}
	\label{relaxationX}
\begin{aligned}
	& (x_{\cdot,k},\YK) \in \mathcal{X}_{\text{L}}(r,k):=\bigg\{ (x_{\cdot, k},\YK)\in\mathbb{R}_{+}^{|\ISET|}\times[0,1]:\\
	&\qquad \qquad\sum\limits_{i \in \ISET} \UIR \XIK \leq \SKR \YK,~\XIK \leq v_{i,k},~\forall~ i \in \ISET ~\bigg \},
\end{aligned}
\end{equation}
making a much larger feasible region for $ (x_{\cdot,k},y_k)$.
Notice that by \eqref{integerX} and $\mathcal{X}(r,k) \subseteq \conv(\mathcal{X}(r,k))$, 
\begin{equation}
	\label{convxrk}
	(x_{\cdot,k},\YK) \in  \conv(\mathcal{X}(r,k)),
\end{equation}
must hold for every feasible solution of problem \eqref{eq:mathforms}.
Therefore, our first refinement of relaxation \eqref{eq:mathforms_lp} is to replace \eqref{relaxationX} with \eqref{convxrk}.
As $ \conv(\mathcal{X}(r,k)) \subseteq \mathcal{X}_{\text{L}}(r,k)$, such a refinement can possibly make a smaller feasible region for $(x_{\cdot,k},\YK) $ when relaxing the integrality requirement on variables $(x_{\cdot,k},\YK) $. 

Next, we pose more restrictions on vector $y$.
In particular, for $r \in \RSET$, adding all the constraints in \eqref{eq:resource} for all $k \in \KSET$ and using constraint \eqref{eq:allocation_general}, 
we obtain  
\begin{equation}
	\label{totalcap}
	\sum\limits_{k \in \KSET}\SKR \YK\geq\sum_{i \in \ISET} \DI \UIR.
\end{equation}
The above constraint requires that the total resources of activated servers should be larger than or equal to the total required resources of all \VMs. 
We remark that in problem \eqref{eq:mathforms}, it is required that 
\begin{equation}
	y \in \mathcal{Y}(r):=\biggl\{ y\in \{0,1\}^{|\KSET|}\,:\,\eqref{totalcap}\biggr\},~\forall~r \in\RSET,  \label{eq:aggregate_con}
\end{equation}
while in relaxation \eqref{eq:mathforms_lp}, $y\in \{0,1\}^{|\KSET|}$ is relaxed to $y\in [0,1]^{|\KSET|}$, and as a result, it follows
\begin{equation}
	y \in \mathcal{Y}_{\text{L}}(r):=\biggl\{ y\in [0,1]^{|\KSET|}\,:\,\eqref{totalcap}\biggr\},~\forall~r \in\RSET. 
\end{equation}
Our second refinement of relaxation \eqref{eq:mathforms_lp} is to enforce 
\begin{equation}
	\label{yconv}
	y \in \conv(\mathcal{Y}(r)),~\forall~r \in\RSET. 
\end{equation}
Similarly, as $\conv(\mathcal{Y}(r)) \subseteq \mathcal{Y}_{\text{L}}(r)$, enforcing \eqref{yconv} in relaxation \eqref{eq:mathforms_lp} can possibly make a smaller feasible region for vector $y$ when relaxing the integrality requirement on variables $y$.

With the above two refinements, we obtain the new relaxation for problem \eqref{eq:mathforms}:
\begin{subequations}
\label{eq:mathforms_lp_es}
\begin{align} 
\mathop{\mathrm{min}} \ 
&\sum_{i \in \ISET}\sum_{k \in \KSET} \CASSIGN \NIK  +\sum_{k \in \KSET} \CRUN \YK 
+ \sum\limits_{i \in \ISET}\sum\limits_{k \in \KSET} \CMIG \ZIK 
\\
\text{s.t.} ~~& 
\eqref{eq:allocation_general}, \eqref{eq:migration1}, \eqref{eq:varX_lp}, \eqref{eq:varY_lp}, & \label{eq:vmcp_polytope} \\
&    (x_{\cdot,k},\YK)  \in \conv(\mathcal{X}(r,k)),~\forall~k \in \KSET, ~\forall~r \in \RSET,  \label{eq:resource_polytope} \\
& y \in \conv(\mathcal{Y}(r)), ~\forall~r \in \RSET. \label{eq:aggregation_polytope}
\end{align}
\end{subequations}
As discussed, \eqref{eq:resource_polytope} and \eqref{eq:aggregation_polytope} make a smaller feasible region for the decision variables in relaxation \eqref{eq:mathforms_lp_es}, as compared with that of relaxation \eqref{eq:mathforms_lp}.
As a result, relaxation \eqref{eq:mathforms_lp_es} can provide a tighter lower bound and hence makes a smaller \CS search tree; {see Section \ref{sect:efficiency_of_cses} further ahead}.

\subsubsection{The cutting plane approach to solve \eqref{eq:mathforms_lp_es}}
\label{sect:framework}

$\conv(\mathcal{X}(r,k))$ and $\conv(\mathcal{Y}(r))$ are polytopes that can be expressed 
by a finite number of inequalities, called \emph{facet-defining inequalities}; see, e.g, \cite[Proposition 8.1]{Wolsey2021}. 
However, it is not practical to solve problem \eqref{eq:mathforms_lp_es} by enumerating 
all inequalities of $\conv(\mathcal{X}(r,k))$ and $\conv(\mathcal{Y}(r))$ due to the following two reasons.
First, it is computationally expensive to find all inequalities that are required to describe $\conv(\mathcal{X}(r,k))$ and $\conv(\mathcal{Y}(r))$.
Second, the numbers of inequalities that describe $\conv(\mathcal{X}(r,k))$ and $\conv(\mathcal{Y}(r))$ 
are potentially huge (usually exponential), making it hard to solve problem \eqref{eq:mathforms_lp_es}.
Due to this, we use a cutting plane approach to solve problem \eqref{eq:mathforms_lp_es}, 
which is used, e.g., in \cite{avella2010computational} in the context of solving the \emph{generalized assignment problem}.
This approach is detailed as follows.
First, we solve relaxation \eqref{eq:mathforms_lp} to obtain its solution $(x^*,y^*)$.
Then, we solve the \emph{separation problem}, that is,
either (i) find a set of inequalities which are valid 
for $ \conv(\mathcal{X}(r,k))$,  $k \in \KSET$ and $r \in \RSET$, and $\conv(\mathcal{Y}(r))$,  $r \in \RSET$, 
but can cut off point $(x^*,y^*)$ (called violated inequalities) or (ii) prove $(x^*,y^*) \in \conv(\mathcal{X}(r,k))$ for all $k \in \KSET$ and $r \in \RSET$  
, and $y^* \in \conv(\mathcal{Y}(r))$ for all $r \in \RSET$.
For case (i), we add the violated inequalities into relaxation \eqref{eq:mathforms_lp} and solve it again. 
For case (ii), $(x^*,y^*)$ must be an optimal solution of problem \eqref{eq:mathforms_lp_es}. 
The above procedure is iteratively applied until case (ii) holds.  
In the following, we demonstrate how to solve the separation problem in detail. \\[5pt]
{\noindent $\bullet$ Integer knapsack set}\\[5pt] 
\indent To solve the separation problem over $\conv(\mathcal{X}(r,k))$ or $\conv(\mathcal{Y}(r))$, it suffices to consider the separation problem over $\conv(\mathcal{X})$, where $\mathcal{X}$ is the generic \emph{integer knapsack set}: 
\begin{equation*}
	\mathcal{X}=\left\{ x \in \mathbb{Z}_{+}^{ |\mathcal{N}|}:\sum_{i \in \mathcal{N}}a_{i}x_{i} \leq b,~x_{i}\leq v_{i},~\forall~i \in \mathcal{N} \right\}, \label{eq:conv-3}
\end{equation*}
$a_i \geq 0$,  $a_i v_i \leq b$ for all $i \in \mathcal{N}$, and $b \geq 0$.
Indeed, by replacing variable $y_{k}$ with $y'_{k}=1-y_{k}$ for all $k \in \mathcal{K}$ in $\mathcal{Y}(r)$, we obtain the so-called \emph{binary knapsack set}
\begin{equation*}
\mathcal{Y}'(r)=\left\{ y^{\prime} \in \left\{0,1\right\}^{\left|\KSET\right|}:\sum_{k \in \KSET}\SKR y'_{k} \leq \sum_{k \in \KSET}s_{k} -\sum_{i \in \ISET}\DI \UIR \right\},~\forall~r \in \RSET, \label{eq:conv-2}
\end{equation*}
which is a special case of integer knapsack set, that is, $v_i =1 $, $i \in \mathcal{N}$ in $\mathcal{X}$.
We remark that inequality $\sum_{k \in \KSET} \alpha_{k,r} y'_k \leq \beta_r$ is valid for $\conv(\mathcal{Y}'(r))$ if and only if $\sum_{k \in \KSET} \alpha_{k,r} (1-y_k) \leq \beta_r$ is valid for $\conv(\mathcal{Y}(r))$.
The set $\mathcal{X}(r,k)$ is a form of  
\begin{equation*}
\mathcal{X}_{y}=\left\{(x,y)\in \mathbb{Z}_{+}^{\left|\mathcal{N}\right|} 
\times \left\{0,1\right\}: \sum_{i \in \mathcal{N}}a_{i}x_{i} \leq by, ~x_{i}\leq v_{i},~\forall~i \in \mathcal{N}\right\}. \label{eq:conv-1} 
\end{equation*}
The following proposition shows that all nontrivial facet-defining inequalities of $\conv(\mathcal{X}_y)$ can be derived from facet-defining  inequalities of $\conv(\mathcal{X})$.
\begin{proposition}
	\label{th:non_trivial_facets}
	(i) All facet-defining inequalities of $\conv(\mathcal{X})$, except $x_i \geq 0$, $i \in \mathcal{N}$, are of the form $\sum_{i \in \mathcal{N}} \pi_i x_i \leq \pi_0$ with $\pi_i \geq 0$, $i \in \mathcal{N}$, and $\pi_0 > 0$; and
	(ii) all facet-defining inequalities of $\conv(\mathcal{X}_{y})$, except $x_i\geq 0$, $i \in \mathcal{N}$, $y \leq 1$, 
	have the form $\sum_{i \in \mathcal{N}}\pi_{i}x_{i} \leq \pi_{0} y$, where $\sum_{i \in \mathcal{N}}\pi_{i}x_{i} \leq \pi_{0}$ is a
	 facet-defining inequality of $\conv(\mathcal{{X}})$ differing from inequalities $x_i \geq 0$, $i \in \mathcal{N}$. 
\end{proposition}
\begin{proof}
	The proof is relegated to the appendix.
\end{proof}
Let $(x^*,y^*) \in \mathbb{R}_+^{|\mathcal{N}|}\times [0,1]$. 
If $y^* =0$, then $x^* = \boldsymbol{0}$ and thus $(x^*,y^*) \in  \conv(\mathcal{X}_y)$.
Otherwise, by Proposition \ref{th:non_trivial_facets}, it follows that $(x^*,y^*) \in \conv(\mathcal{X}_y)$ if and only if $\bar{x} \in \conv(\mathcal{X})$ where $\bar{x}_i = \frac{x^*_i}{y^*}$ for all $i \in \mathcal{N}$.
Based on the above discussion, we shall concentrate on the separation problem over $\conv(\mathcal{X})$ in the following.
\\[5pt]
{\noindent $\bullet$ Exact separation for integer knapsack polytope $\conv(\mathcal{X})$}\\[5pt] 
\indent Next, we solve the separation problem of polytope $\conv(\mathcal{X})$, that is, either construct a hyperplane separating point $\bar{ x}$ from the $\conv(\mathcal{X})$ strictly, i.e.,
\begin{align}
	&\pi^{\top} x \leq \pi_{0} , ~\forall~x \in \conv(\mathcal{X})~ (\text{or equivalently},~\forall~x \in\mathcal{X}), \label{eq:hyperplane}
\end{align}
and 
\begin{align}
	& { \bar{x}}^{\top}\pi > \pi_{0}, \label{eq:violate}
\end{align}
or prove that none exists, i.e., $\bar{x} \in \conv(\mathcal{X})$.
This separation problem can be reduced to the following \LP problem
\begin{equation}
	\label{eq:separation_optimization}
	\omega(\bar{x}) = \max_{(\pi, \pi_0) \in \mathbb{R}^{|\mathcal{N}|+1}}\left\{{\bar{x}}^{\top}\pi - \pi_{0} : 
	\pi^{\top} x  \leq \pi_0,~\forall~x \in \mathcal{X} \right\}, 
\end{equation}
If $\omega( \bar{x}) \leq 0$, we must have  $\bar{x} \in \conv(\mathcal{X})$;
otherwise, $\pi^{\top}x \leq \pi_{0}$ is a valid inequality violated by $\bar{x}$.
By Proposition \ref{th:non_trivial_facets} (i), we can, without loss of generality, add $\pi_i \geq 0$ for all $i \in \mathcal{N}$ and $\pi_0 >0$ into problem \eqref{eq:separation_optimization}.
Moreover, we can further normalize $\pi_0$ as $1$ (as $\pi_0 >0$) and obtain the following equivalent problem
\begin{equation}
\omega(\bar{x})  = \max_{\pi \in \mathbb{R}^{|\mathcal{N}|}_+}\left\{{ \bar{x}}^{\top} \pi : 
 \pi^{\top} x  \leq 1,~\forall~x \in \mathcal{X} \right\}. \label{eq:sepa_problem}
\end{equation}
In particular, letting $\pi^*$ be an optimal solution of \eqref{eq:sepa_problem}, if ${ \bar{x}}^{\top}\pi^{*} \leq 1$, we prove  $\bar{x} \in \conv(\mathcal{X})$;
otherwise, we find the inequality ${\pi^{*}}^{\top}x \leq 1$ violated by $\bar{x}$.
One weakness of problem \eqref{eq:sepa_problem} is its large problem size. 
Indeed, the number of constraint $\pi^\top x \leq 1$ may be exponential as the points in $\mathcal{X}$ may be exponential.
Consequently, from a computational perspective, it is not practical to solve the separation problem when all constraints are explicitly expressed. 
For this reason, we use the \emph{row generation method}, an iterative approach that starts with a subset of constraints 
and then dynamically adds other constraints when violations occur \cite{vasilyev2016implementation}.
More specifically, we first choose an initial subset
\begin{equation}
	\mathcal{U}= \left\{v_i \boldsymbol{e}^{i} \, : \, i \in \mathcal{N} \right\} \subseteq \mathcal{X}, \label{eq:initial_set}
\end{equation}
where $\boldsymbol{e}^{i}$ is the $i$-th $\left|\mathcal{N}\right|$-dimensional unit vector, 
and solve the partial separation problem defined by this subset $\mathcal{U}$: 
\begin{equation}
	\omega'(\bar{x}) =\max_{{\pi \in \mathbb{R}^{|\mathcal{N}|}_+}} \left\{ { \bar{x}}^{\top} \pi : 
	 \pi^{\top}x \leq 1,~\forall~{x} \in \mathcal{U} \right\}. \label{eq:partial_sepa_problem}
\end{equation}
Let $\pi^*$ be an optimal solution of \eqref{eq:partial_sepa_problem}.
If $\omega'\left(\bar{x}\right) \leq 1$, we prove  $\bar{x} \in \conv(\mathcal{X})$ (as $\mathcal{U} \subseteq \mathcal{X}$);
otherwise, we check whether ${\pi^*}^\top x \leq 1$ holds for all $x \in \mathcal{X}$ by solving the following bounded knapsack problem
\begin{equation}
	{h}^{*} \in \argmax_{{h}} \left\{ {\pi^{*}}^{\top}h : \forall~{h} \in \mathcal{X}\right\}. \label{eq:knap_problem}
\end{equation}
(i) If ${\pi^{*}}^{\top}{h^{*}} \leq 1$,  ${\pi^{*}}^{\top}x \leq 1$ holds for all $x \in \mathcal{X}$;
(ii) otherwise, we add $h^{*}$ into $\mathcal{U}$ and solve problem \eqref{eq:partial_sepa_problem} again.
The above procedure is iteratively applied until case (i) holds.
The row generation method is summarized in Algorithm \ref{alg:row_generation}. 
\begin{algorithm}[!htbp]
	\caption{The row generation method to solve the separation problem of $\conv(\mathcal{X})$.} 
	\label{alg:row_generation}
	\KwIn{The set $\mathcal{X}$ and a solution ${\bar{x}}$.}
	\KwOut{Find a violated inequality ${\pi^{*}}^{\top}x \leq 1$ separating ${ \bar{x}}$ from $\conv(\mathcal{X})$ 
		or conclude that ${ \bar{x}} \in \conv(\mathcal{X})$.}
	\BlankLine
	Choose an initial subset $\mathcal{U}$ as in \eqref{eq:initial_set};\\
	Solve the partial separation problem \eqref{eq:partial_sepa_problem} to obtain its solution $\pi^*$;\\
	If ${ \bar{x}}^{\top} \pi^{*} \leq 1$, 
	conclude ${\bar{x}} \in \conv(\mathcal{X})$ and stop; 
	otherwise, solve the bounded knapsack problem \eqref{eq:knap_problem} to obtain the solution $h^*$; \\
	If ${\pi^{*}}^{\top}{h^{*}}  > 1$, set $\mathcal{U} \coloneqq \mathcal{U} \cup \{{h}^{*}\}$ and go to step 2;	otherwise, stop and return the violated inequality ${\pi^{*}}^{\top}x \leq 1$;
\end{algorithm}

We remark that problem \eqref{eq:knap_problem} is an integer knapsack problem which is generally NP-hard. 
However, it can be solved by the dynamic programming algorithm in \cite{pisinger2000minimal}, which runs in pseudo-polynomial time but is quite efficient in practice.
In addition, problem \eqref{eq:knap_problem} may have multiple optimal solutions. 
We follow \cite{chen2021exact} to choose a solution with large values $h^*$.
This strategy provides a much stronger inequality ${\pi}^\top h^* \leq 1$ for problem \eqref{eq:partial_sepa_problem}, which effectively decreases the number of iterations in Algorithm \ref{alg:row_generation}. 
For more details, we refer to \cite{chen2021exact}.\\[5pt]
{\noindent $\bullet$~Efficient implementation}\\[5pt] 
\indent In each iteration of Algorithm \ref{alg:row_generation}, we need to solve the \LP problem \eqref{eq:partial_sepa_problem}, which is still time-consuming, especially when the dimension of $\mathcal{X}$ is large.
Below, we introduce two simple techniques to reduce the dimension of $\mathcal{X}$.

First, we can aggregate multiple variables with the same coefficient into a single variable. 
Specifically, suppose that $a_i$, $i \in \mathcal{N}'$, are equal.
Then replacing $\sum_{i \in \mathcal{N}'} x_i $ by a new variable $\delta$ (with $\delta \leq \sum_{i \in \mathcal{N}^{\prime}} v_i$) in  $\mathcal{X}$, we obtain a new set $\mathcal{X}'$.
After constructing a valid inequality for $\mathcal{X}'$, we substitute $\delta = \sum_{i \in \mathcal{N}^{\prime}}x_i$ into this inequality and obtain a valid inequality for $\mathcal{X}$.
In our experience, this simple technique effectively reduces the CPU time spent by Algorithm \ref{alg:row_generation}, especially when most coefficients in $\mathcal{X}$ are equal.
The second technique to reduce the dimension of $\mathcal{X}$ comes from Vasilyev et al. \cite{vasilyev2016implementation}, which consists of two steps.
In the first step, the separation problem over a projected polytope $\conv(\mathcal{X}({\bar{x}}))$ is solved to obtain a valid inequality (violated by $\bar{x}$):
\begin{equation}
	 \label{valid_inequality}
	\sum_{i \in \mathcal{N}_{\rm R}}\pi_{i}x_{i} 
	\leq 1,
\end{equation}
 where 
\begin{align}
	\mathcal{X}({\bar{x}})
	=\left\{{x}\in \mathbb{Z}_{+}^{\left|\mathcal{N}_{\rm R}\right|} : 
	\sum_{i \in \mathcal{N}_{\rm R}  } a_{i}x_{i} \leq \bar{b},~x_{i} \leq v_{i},~\forall~i \in \mathcal{N}_{\rm R} \right\},
\end{align}
$\mathcal{N}_{\rm L}=\left\{i \in \mathcal{N}:\bar{x}_{i}=0\right\}$, 
$\mathcal{N}_{\rm U}=\left\{i \in \mathcal{N}:\bar{x}_{i}=v_{i}\right\}$, 
$\mathcal{N}_{\rm R}=\mathcal{N} \backslash (\mathcal{N}_{\rm L} \cup \mathcal{N}_{\rm U} )$, and $\bar{b}=b - \sum_{i \in \mathcal{N}_{\rm U}}a_{i}v_{i}$.
It can be expected that the separation problem over $\conv(\mathcal{X}({\bar{x}}))$ is easier to solve than that over $\conv(\mathcal{X})$, especially when the number of fixed variables $|\mathcal{N}_{\text{L}}|+|\mathcal{N}_{\text{U}}|$ is large. 
In the second step, we derive a valid inequality 
\begin{align}
	&\sum_{i \in \mathcal{N}_{\rm R}}\pi_{i}x_{i} 
	+ \sum_{i \in \mathcal{N}_{\rm L}} {\pi}_{i}x_{i}
	+ \sum_{i \in \mathcal{N}_{\rm U}} {\pi}_{i}x_{i}
	\leq 1 + \sum_{i \in \mathcal{N}_{\rm U}} {\pi}_{i}v_{i}, \label{eq:lift_valid_inequality}
\end{align}
for $\conv(\mathcal{X})$ using \emph{sequential lifting}; see, e.g.,  \cite{kaparis2010separation,vasilyev2016implementation,gu1998lifted,gu2000sequence,zemel1989easily} for a detailed discussion of sequential lifting.
\section{Numerical results}
\label{sect:numerical_results}

In this section,  we present simulation results to illustrate the effectiveness and efficiency of 
the proposed formulation \eqref{eq:mathforms} and the proposed \CSES algorithm for solving \VMCPs.
More specifically, we first perform numerical experiments to compare the performance of solving the proposed formulation \eqref{eq:mathforms} and the two existing formulations in \cite{speitkamp2010mathematical} and \cite{mazumdar2017power} by standard \MIP solvers. 
Then, we present some simulation results to demonstrate the efficiency of the proposed \CSES algorithm 
for solving \VMCPs over standard \MIP solvers.
Finally, we evaluate the performance of our proposed \CS algorithm under different problem parameters.

In our implementation, the proposed \CSES was implemented in C++ linked with IBM ILOG \cplex optimizer 20.1.0 \cite{CPLEX}.
The time limit and relative gap tolerance were set to 7200 seconds and 0\%, respectively, in all experiments.
The cutting plane approach was stopped if the optimal value of the \LP relaxation problem of \VMCP is improved 
by less than 0.05\% between two adjacent calls.
Unless otherwise specified, all other \cplex parameters were set to their default values.
All experiments were performed on a cluster of Intel(R) Xeon(R) Gold 6140 @ 2.30GHz computers, with 192 GB RAM, running Linux (in 64 bit mode).

\subsection{Testsets}
We tested all algorithms on problem instances with 5 \VM types and 10 server types with different features (CPU, RAM, and Bandwidth resources and power consumption), as studied in \cite{mazumdar2017power}. 
The \VM types and server types are in line with industry standards and are described in Tables \ref{table:first_dataset_vm} and \ref{table:first_dataset_server}, respectively.

The basic \VMCP \eqref{eq:mathforms} instances are constructed using the same procedure in \cite{mazumdar2017power}.
Specifically, each instance has an equal number of servers of each type.
The number of servers $|\KSET|$ is selected from $\{ 250, 500, 750,1000 \}$.
For each server $k$, we iteratively assign a uniform random number $\NIK$ (satisfying $\NIK \in \{0,\ldots,\lfloor\frac{\SKR}{\UIR}\rfloor\}$) of \VMs of type $i$ until the maximum usage of the available resource (CPU, RAM, and Bandwidth) load $\sigma_{k}$, defined by,
\begin{equation}
	\sigma_{k}=\max\biggl\{\frac{\sum_{i \in \ISET}\UIR \NIK}{\SKR}: \forall~r \in \RSET \biggr\}, \label{eq:generation_instance}
\end{equation}
exceeds a predefined value $\alpha$.
In general, the larger the $\alpha$, the more \VMs will be constructed.
We choose $\alpha \in \{20\%, 40\%\}$.
In our test, we attempt to minimize the total power consumption of the servers.
As shown in  \cite{mazumdar2017power,gao2013multi,wu2016energy}, the power consumption of a server can be represented by the following linear model:
\begin{align} 
	\label{eq:power_consumption_server}
	P_{k}=P_{{\rm idle},k} +\left(P_{{\rm max},k}-P_{{\rm idle},k}\right)U_{k},
\end{align}
where $P_{{\rm idle},k}$ is the idle power consumption (at the idle state) of server $k$, 
$P_{{\rm max},k}$ is the maximum power consumption (at the peak state) of server $k$, and 
$U_{k}$ ($U_{k}\in \left[0,1\right]$) is the CPU utilization of server $k$.
As such, (i) the activation cost $\CRUN$ is set to the idle power consumption of server $k$, which is equal to 60\% of the maximum power consumption, as assumed by \cite{mazumdar2017power}; and (ii) the allocation cost $\CASSIGN$ is set to $\left(P_{{\rm max},k}-P_{{\rm idle},k}\right)\frac{u_{i,{\rm CPU}}}{s_{k,{\rm CPU}}}$.
As illustrated in Section \ref{sect:problem_formualtion}, $\CMIG$ is set to $\CASSIGN$ in the experiments.

The extended \VMCP (i.e., problem \eqref{eq:mathforms} with constraints \eqref{eq:new_imcoming_vms}-\eqref{eq:attribute1}) instances are constructed based on the basic \VMCP instances described above with the parameters for constraints \eqref{eq:new_imcoming_vms}-\eqref{eq:attribute1} described as follows.
{For constraints  \eqref{eq:new_imcoming_vms} and  \eqref{eq:transformed_resource}, parameter $\DINEW$ is obtained by setting existing \VMs as new incoming \VMs with a probability $\beta$, chosen in \{35\%, 45\%\}.
}
For constraint \eqref{eq:limit_number_vm_migrations}, parameter $\L$ is set as  $\eta \sum\limits_{i \in \ISET}\sum \limits_{k \in \KSET}\NIK $ where $\eta$ is chosen in $\left\{30\%, 40\% \right\}$.
{For constraints \eqref{eq:max_number_vms} or \eqref{eq:max_number_vms1}, the maximum number of \VMs on each server $k$, $\MK$, is set to $\lambda \max\limits_{i \in \ISET} \min \limits_{r \in \RSET} \lfloor \frac{\SKR}{\UIR} \rfloor$, where 
$\lambda \in \left\{85\%, 90\%\right\}$.
Notice that $\min \limits_{r \in \RSET} \lfloor\frac{\SKR}{\UIR} \rfloor$ is the maximum number of \VMs of type $i$ that can be allocated to a given server $k$, and hence $\max\limits_{i \in \ISET} \min \limits_{r \in \RSET} \lfloor\frac{\SKR}{\UIR}\rfloor$ is the maximum number of \VMs of a single type that can be allocated to a given server $k$.}
For constraints \eqref{eq:attribute} and \eqref{eq:attribute1}, 
we randomly choose the element $k \in \KSET$ to subset $\KSET(i)$, $i \in \ISET$, with a probability $\theta$, chosen in $\{5\%, 10\%\}$.

For each fixed $|\mathcal{K}| \in \{250,500,750,1000\}$ and $\alpha\in \{20\%, 40\%\}$, 50 basic \VMCP instances are randomly generated, leading to an overall 400 basic \VMCP instances testbed.
In addition, for each fixed $|\mathcal{K}| \in \{250,500,750,1000\}$, $\alpha\in \{20\%,40\%\}$, $\beta \in \{35\%, 45\%\}$, $\eta \in \{30\%, 40\%\}$, $\lambda \in \{85\%, 90\%\}$, and $\theta \in \{5\%, 10\%\}$, 10 extended \VMCP instances are randomly generated, leading to an overall $1280$ extended \VMCP instances testbed.

\begin{table}[t]
	\centering
	\setlength{\tabcolsep}{10pt}
	\renewcommand{\arraystretch}{1.1}
	\caption{The five \VM types.}
	\label{table:first_dataset_vm}
	\small
	\begin{tabular}{cccc}
		\hline
		Type&CPU&RAM (GB)&Bandwidth (Mbps)\\
		\hline
		\VM 1&1&1&10 \\
		\VM 2&2&4&100\\
		\VM 3&4&8&300\\
		\VM 4&6&12&1000\\
		\VM 5&8&16 &1200\\
		\hline
	\end{tabular}
\end{table}

\begin{table}[t]
	\centering
	\setlength{\tabcolsep}{5pt}
	\renewcommand{\arraystretch}{1.1}
	\caption{The ten server types.}
	\label{table:first_dataset_server}
	\small
	\begin{tabular}{ccccc}
		\hline
		\multirow{1.5}{*}{Type}&\multirow{1.5}{*}{CPU}&\makecell*[c]{RAM\\(GB)} &\makecell*[c]{Bandwidth\\(Mbps)}&\makecell*[c]{Maximum power \\consumption (W)}\\
		\hline
		Server 1&4&8 &1000 &180\\
		Server 2&8&16 & 1000&200\\
		Server 3&10&16&2000&250\\
		Server 4&12&32&2000&250\\
		Server 5&14&32&2000&280\\
		Server 6&14&32&2000&300\\
		Server 7&16&32&4000&300\\
		Server 8&16&64&4000&350\\
		Server 9&18&64&4000&380\\
		Server 10&18&64&4000&410\\
		\hline
	\end{tabular}
	\vspace{-15pt}
\end{table}
\subsection{Efficiency of the proposed formulation}
\label{sect:comparison_with_other_two_mathematical_formulations}
In this subsection, we present computational results to illustrate the efficiency of the proposed formulation
\eqref{eq:mathforms} over those in \cite{speitkamp2010mathematical} and \cite{mazumdar2017power} (i.e., formulations \eqref{eq:mathforms_sb} and \eqref{eq:mathforms_int}).
Table \ref{table:comparison_results_three_model_5} summarized  the computational results of the three formulations solved by \cplex.
We report the number of instances that can be solved to optimality witnin the given time limit (\Solved), the average CPU time (\TT), and the average number of constraints and variables (\NC and \NV, respectively).  

\begin{table*}[htbp]
	\centering
	\begin{threeparttable}[b]
		\small
		\caption{Comparison results of the proposed formulation \eqref{eq:mathforms} with existing formulations in \cite{speitkamp2010mathematical} and \cite{mazumdar2017power}}
		\label{table:comparison_results_three_model_5}
		\renewcommand{\arraystretch}{1.1}
		\setlength{\tabcolsep}{4.5pt}
		\begin{tabular}{llllllllllllll}
			\hline
			\multirow{3}{*}{$\left(\left|\mathcal{K}\right|, \alpha\right)$}& \multirow{3}{*}{\NVM} &\multicolumn{4}{c}{Proposed formulation \eqref{eq:mathforms}}& \multicolumn{4}{c}{\makecell*[c]{Formulation in \\Speitkamp and Bichler \cite{speitkamp2010mathematical}}} &\multicolumn{4}{l}{\makecell*[c]{Formulation in \\Mazumdar and Pranzo \cite{mazumdar2017power}}} \\
			\cmidrule(r){3-6}\cmidrule(r){7-10}\cmidrule(r){11-14}
			& &\Solved & \TT & \NC & \NV&\Solved & \TT & \NC & \NV&\Solved & \TT & \NC & \NV \\
			\hline
			(250, 20\%) & 959& \textbf{50} & \textbf{1.5}  & 2005  & 2750&46 & 905.8 & 241459     & 479750&45 & 574.3 & 2000     & 312750 \\
			(250, 40\%) & 1163& \textbf{49} & \textbf{2.9} & 2005 & 2750 &37 & 1664.7& 292663     & 581750& 37 & 870.7 & 2000    & 312750 \\
			(500, 20\%) &1924& \textbf{50} & \textbf{3.9} & 4005 & 5500& 20 & 5179.7& 965424     & 1924500& 20 & 3782.5& 4000    & 1250500  \\
			(500, 40\%) &2310& \textbf{45} & \textbf{12.5}  & 4005 & 5500 & 8  & 5266.9& 1158810  & 2310500 &13 & 4302.6& 4000    & 1250500 \\
			(750, 20\%) &2855 & \textbf{41} & \textbf{38.4}  & 6005 & 8250 & 1  & 6996.1& 2146355    & 4283250& 1  & 6962.0& 6000        & 2813250    \\
			(750, 40\%) & 3493 & \textbf{31} & \textbf{108.7} & 6005  & 8250& 1  & 6998.3& 2625493     & 5240250& 3  & 6825.0& 6000        & 2813250  \\
			(1000, 20\%)&3832 & \textbf{36} & \textbf{67.7}  & 8005 & 11000& 0  & -     &3838832     & 7665000& 0  & -     &8000     & 5001000   \\
			(1000, 40\%)&4655 & \textbf{23} & \textbf{316.2} & 8005  & 11000 &2  & 7038.0& 4662655     & 9311000& 1  & 7198.1& 8000      & 5001000  \\
			\hline
			All & &\textbf{325}&\textbf{19.5}&&&115&4227.8& && 120&3449.4& & \\ 
			\hline& 
		\end{tabular}
			\vspace{-15pt}
		\begin{tablenotes}
			\item In the table, ``-'' means that the average CPU time reaches the time limit.
		\end{tablenotes}
	\vspace{-15pt}
\end{threeparttable}
\end{table*}

As expected, (i) the numbers of variables and constraints in the proposed formulation \eqref{eq:mathforms} are much smaller than those in \cite{speitkamp2010mathematical}; and (ii) the number of variables in the proposed formulation \eqref{eq:mathforms} is also much smaller than the one in \cite{mazumdar2017power}, but the numbers of constraints in the two formulations are fairly equal.
Consequently, it can be clearly seen that it is much more efficient to solve formulation \eqref{eq:mathforms} than those in \cite{speitkamp2010mathematical} and \cite{mazumdar2017power}.
More specifically, using the proposed formulation \eqref{eq:mathforms}, 325 instances (among 400 instances) can be solved to optimality. 
In sharp contrast, using the formulations in \cite{speitkamp2010mathematical} and \cite{mazumdar2017power}, only 115 and 120 instances can be solved to optimality, respectively. 
Indeed, for large-scale cases (e.g., $|\mathcal{K}|=750, 1000$), only a few instances can be solved to optimality by using the two existing formulations in \cite{speitkamp2010mathematical} and \cite{mazumdar2017power}.
Moreover, as observed in Table \ref{table:comparison_results_three_model_5}, compared with the formulations in \cite{speitkamp2010mathematical} and \cite{mazumdar2017power}, the CPU time taken by solving the proposed formulation \eqref{eq:mathforms} is much smaller (19.5 seconds versus 4227.8 seconds and 3449.4 seconds).
From this computational result, we can conclude that formulation \eqref{eq:mathforms} significantly outperforms the formulations in \cite{speitkamp2010mathematical} and \cite{mazumdar2017power} in terms of solution efficiency.  

\subsection{Efficiency of the proposed $\CS$ algorithm}
\label{sect:efficiency_of_cses}
In this subsection, we compare the performance of the proposed \CSES algorithm with the approach using the \MIP solver \cplex (called \DEF).
In addition, to address the advantage of embedding the proposed relaxation \eqref{eq:mathforms_lp_es} into the \CS algorithm, 
we compare \CSES with \CSP algorithm, in which the LP relaxation \eqref{eq:mathforms_lp} is used, to solve \VMCPs. 

\begin{figure}[htbp]
	\centering
	\includegraphics[width=0.8\linewidth]{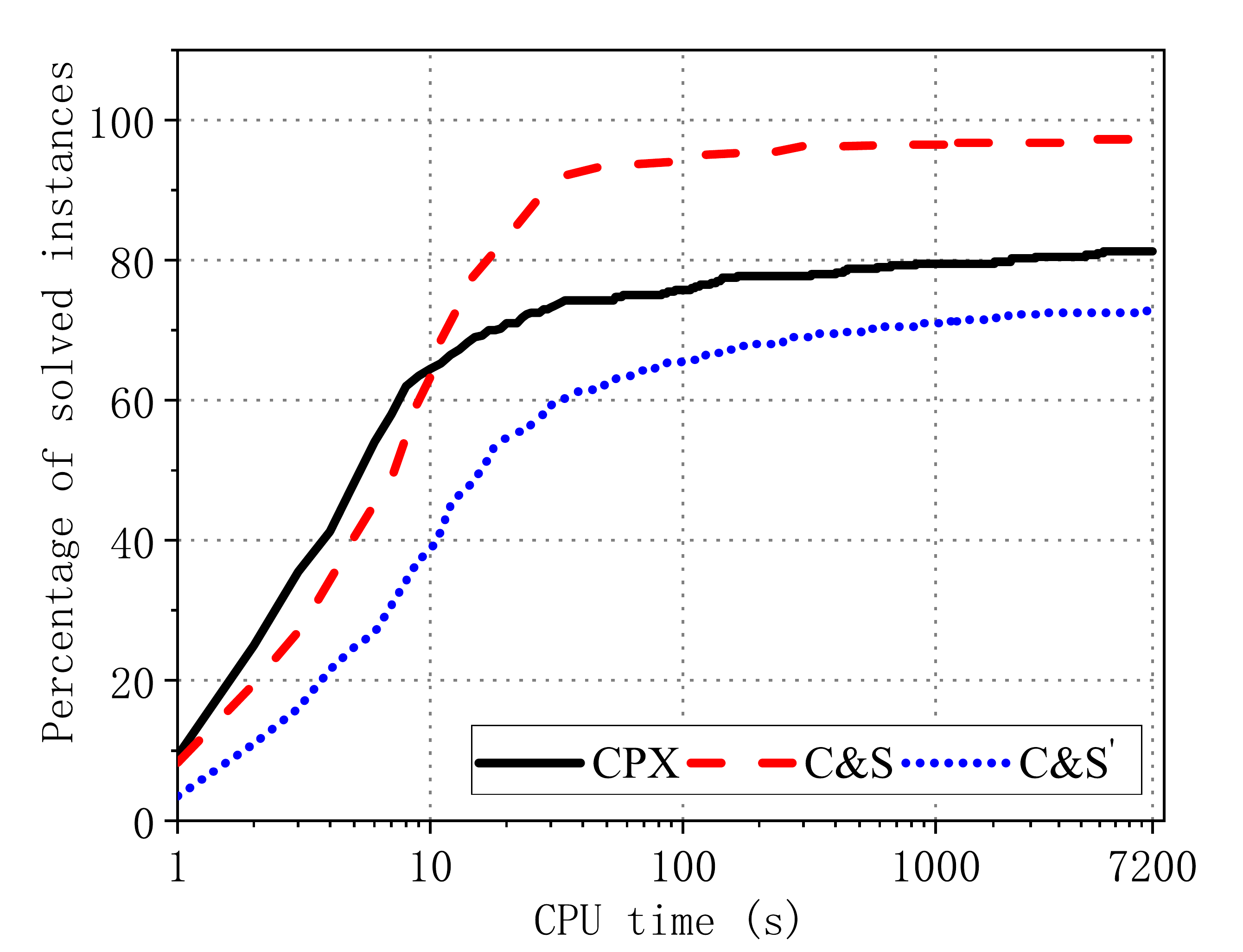}
	\caption{Comparison of the CPU time between \DEF, \CSES, and \CSP on basic \VMCPs.}  
	\label{fig:time-all-1-1}
\end{figure}

Figs. \ref{fig:time-all-1-1}-\ref{fig:time-all-2-1} plot the  performance profiles of the three settings \CPX, 
\CSES, and \CSP.
Each point with coordinates $(a, b)$ in a line represents that for $b\%$ of the instances, 
the CPU time is less than or equal to $a$ seconds.
From Figs. \ref{fig:time-all-1-1}-\ref{fig:time-all-2-1}, \DEF can solve more basic \VMCP instances within 10 seconds 
and more extended \VMCP instances within 5 seconds, respectively. 
This shows that \DEF performs a bit better than \CSES for easy instances.
However, for the hard instances, \CSES significantly outperforms \DEF, especially for basic \VMCPs. 
In particular, \CSES can solve 97\% basic \VMCP instances to optimality while \DEF can solve only 81\% basic \VMCP instances to optimality.
\begin{figure}[t]
	\centering
	\includegraphics[width=0.8\linewidth]{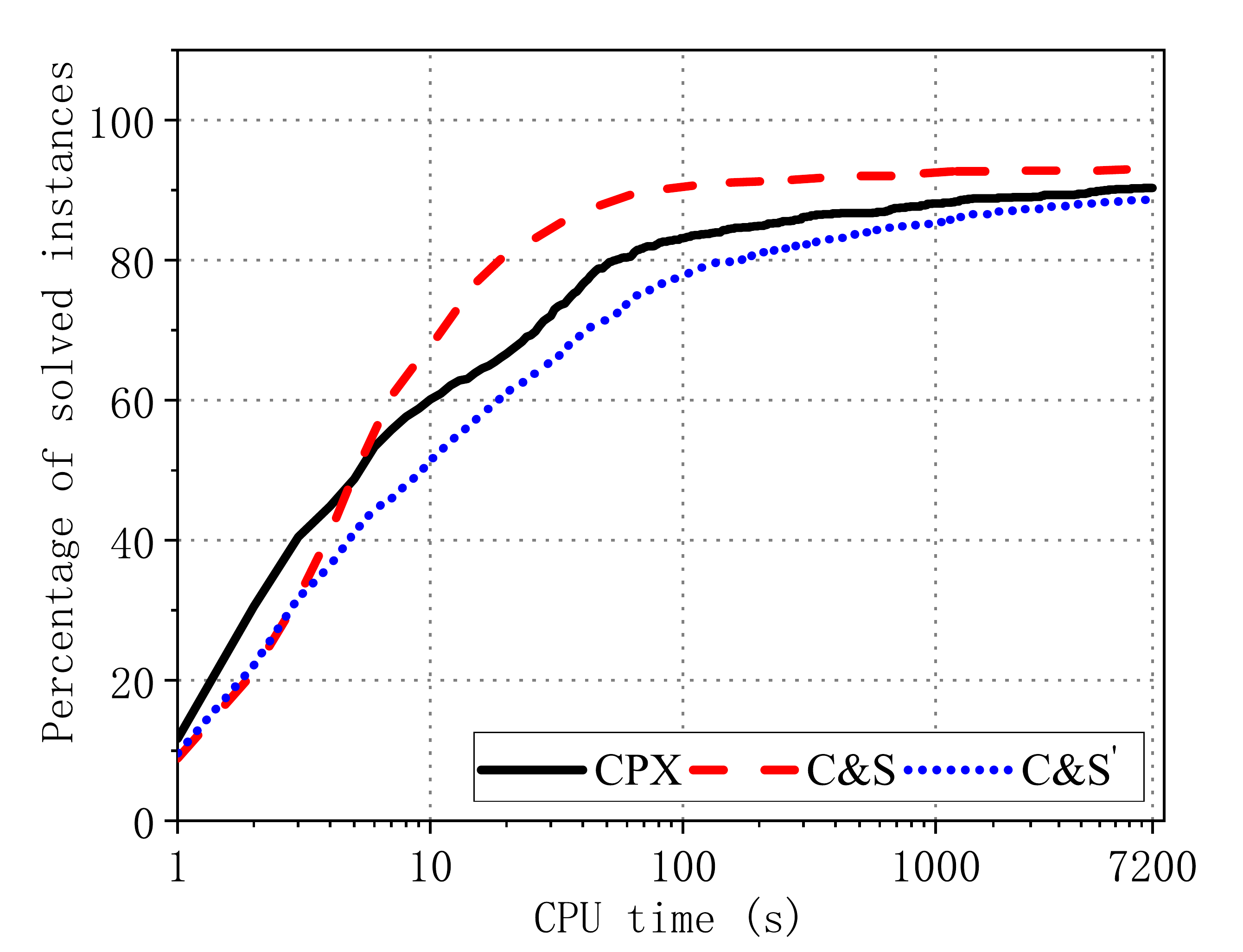}
	\caption{Comparison of the CPU time between \DEF, \CSES, and \CSP on extended \VMCPs.}
	\label{fig:time-all-2-1}
\end{figure}
\begin{figure}[t]
	\centering
	\includegraphics[width=0.8\linewidth]{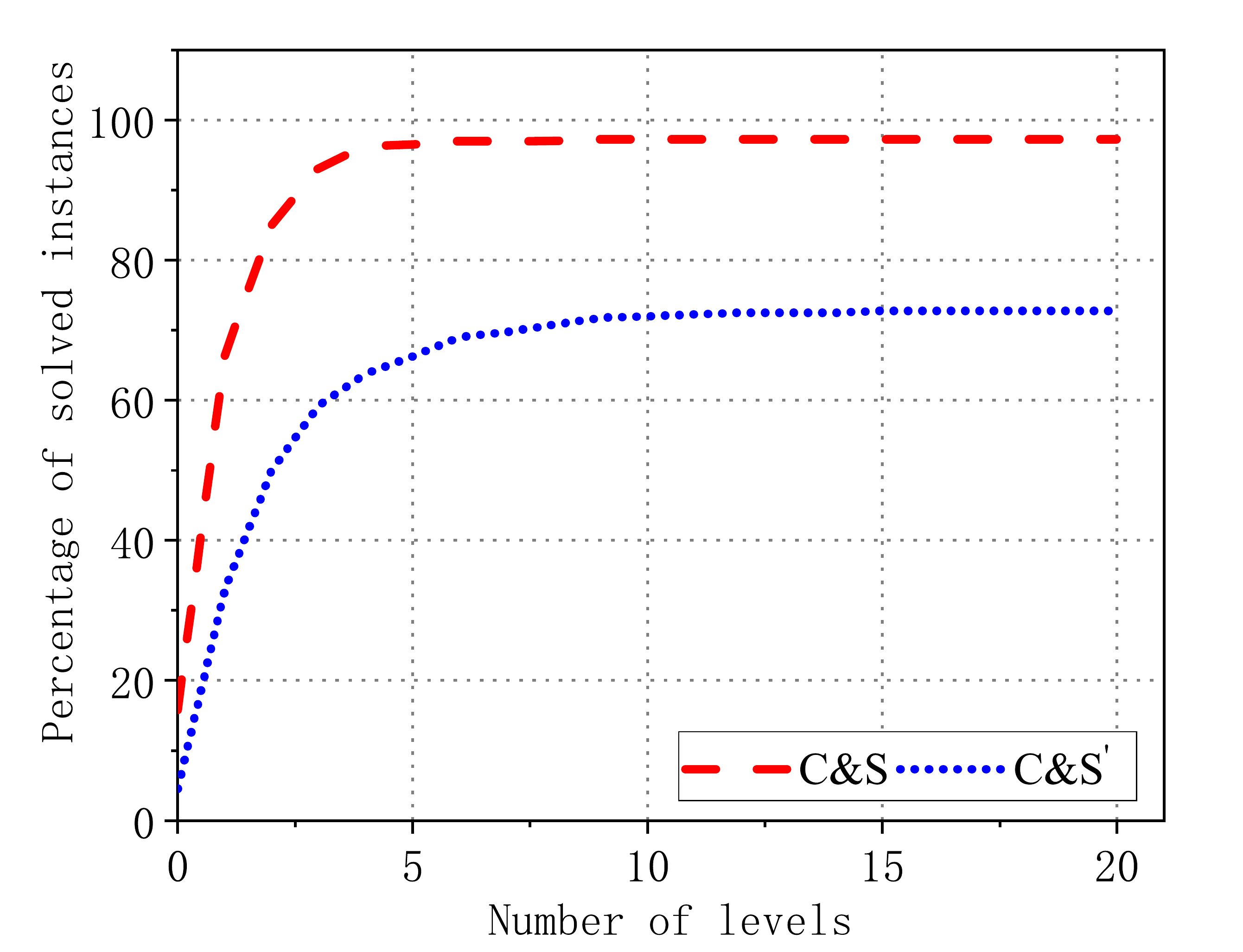}
	\caption{Comparison of the number of levels of the search tree between \CSES and \CSP on basic \VMCPs.} 
	\label{fig:level-all-1-1}
\end{figure}
In addition, from the two figures, we can conclude that the performance of 
\CSES is much better than \CSP for basic and extended \VMCPs. 
This indicates that the proposed compact relaxation \eqref{eq:mathforms_lp_es} has a significantly positive performance impact on the \CS approach.

To gain more insight into the computational efficiency of \CSES over \CSP, 
we compare the numbers of levels of the cut-and-solve search trees returned by \CSES and \CSP.
\begin{figure}[t]
	\centering
	\includegraphics[width=0.8\linewidth]{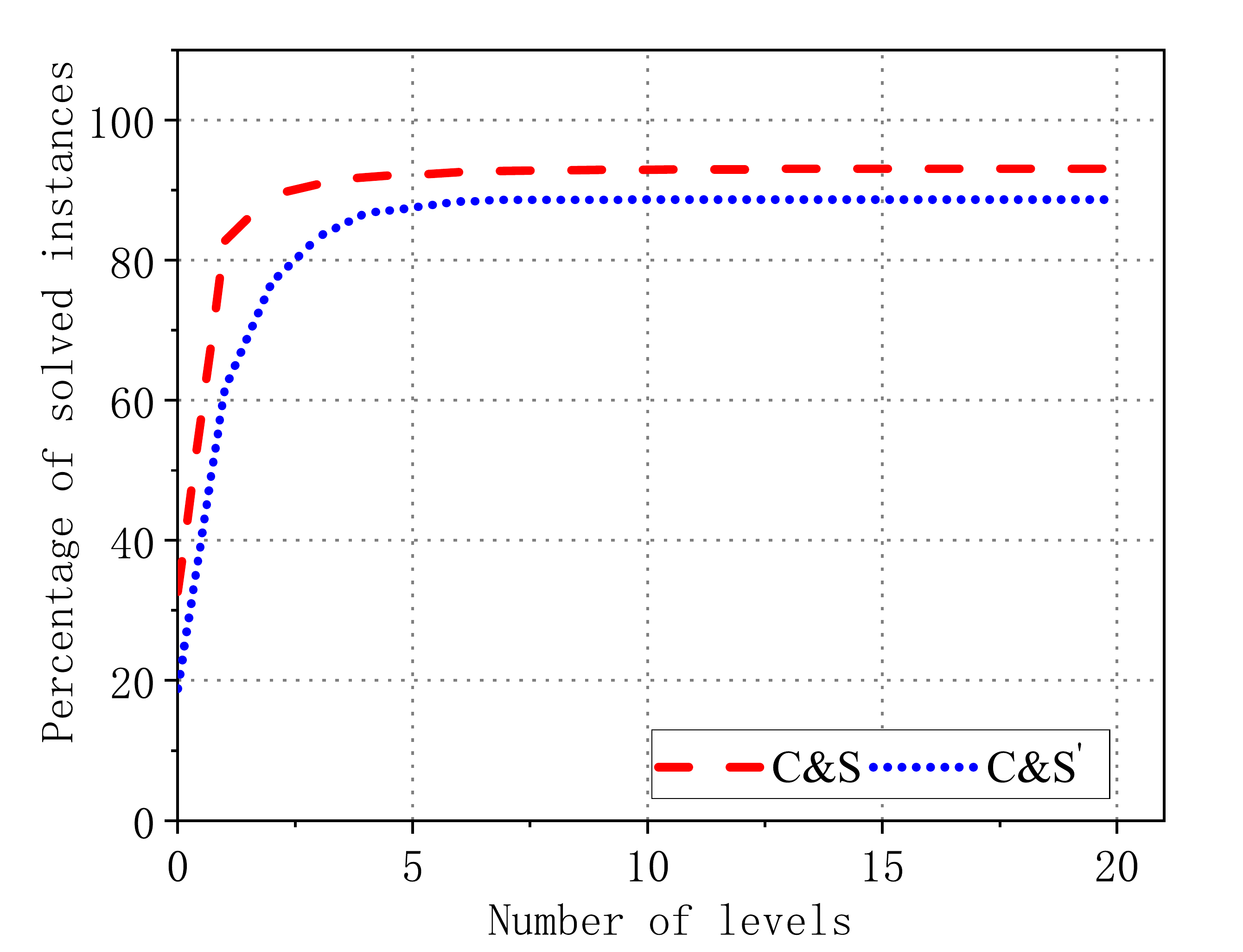}
	\caption{Comparison of the number of levels of the search tree between \CSES and \CSP on extended \VMCPs.}
	\label{fig:level-all-2-1}
\end{figure}
The results for the basic and extended \VMCP instances are summarized in Figs.  \ref{fig:level-all-1-1} and \ref{fig:level-all-2-1}, respectively.
From the two figures, we can conclude that the number of levels returned by \CSES is much less than that returned by \CSP, 
especially for basic \VMCP instances. 
More specifically, more than 97\% of the basic \VMCP instances can be solved by \CSES within 5 levels 
while only about 70\% of the basic \VMCP instances can be solved by \CSP within 20 levels.
This shows the advantage of the proposed relaxation \eqref{eq:mathforms_lp_es}, i.e., it can effectively reduces the \CS search tree size.

From the above computational results, we can conclude that 
(i) the proposed \CSES algorithm is much more effective than standard \MIP solver, especially for the hard instances; 
(ii) the proposed relaxation \eqref{eq:mathforms_lp_es} can effectively reduce the size of the \CS search tree, 
which plays a crucial role in the efficiency of the proposed \CSES algorithm. 

\subsection{Performance comparison of the proposed \text{\rm{\CSES}} algorithm}
\label{sect:performance_comparison}
To gain more insights into the performance of the proposed \CSES algorithm,  we compare the performance of \CSES on instances with different numbers of servers $\left|\mathcal{K}\right|$ and different loads $\alpha$ (the higher the load $\alpha$, the larger the number of \VMs).
\begin{figure}[t]
	\centering
	\includegraphics[width=0.8\linewidth]{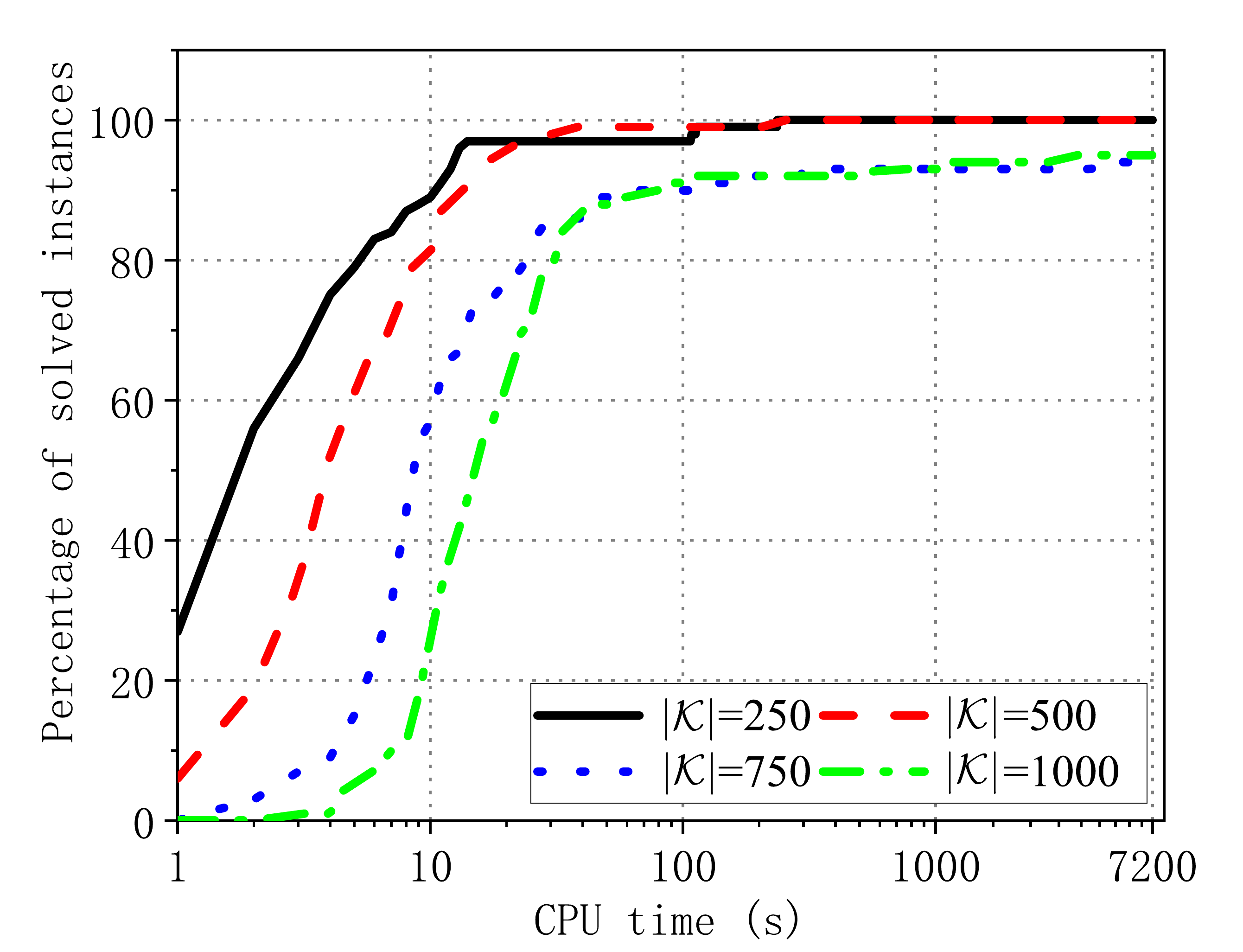}
	\caption{Comparison of the CPU time for \CSES on basic \VMCPs with different numbers of servers $\left|\mathcal{K}\right|$.} 
	\label{fig:time-k-1-1}
\end{figure}
\begin{figure}[t]
	\centering
	\includegraphics[width=0.8\linewidth]{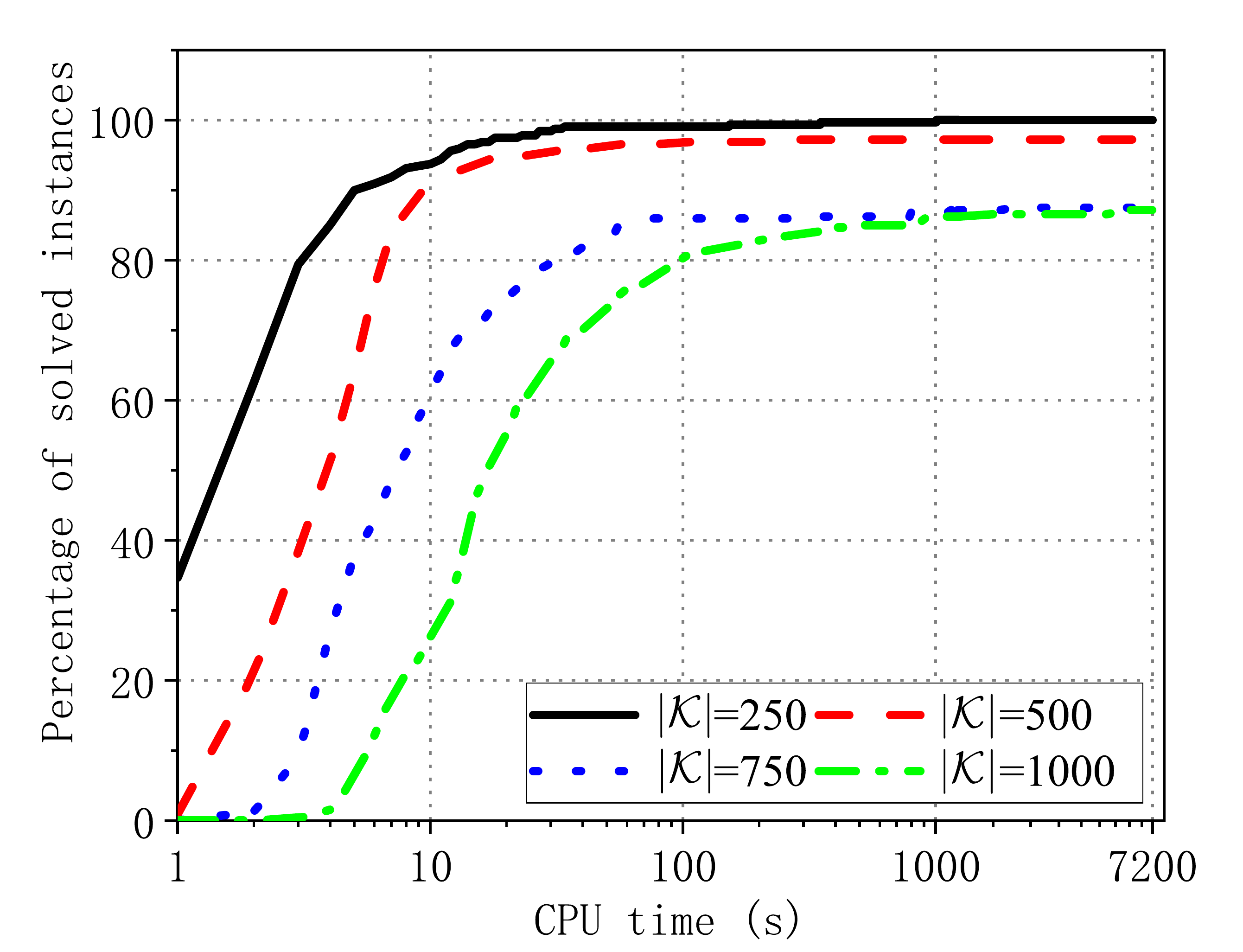}
	\caption{Comparison of the CPU time for \CSES on extended \VMCPs with different numbers of servers $\left|\mathcal{K}\right|$.} 
	\label{fig:time-k-2-1}
\end{figure}

Figs. \ref{fig:time-k-1-1} and \ref{fig:time-k-2-1} plot performance profiles of CPU time, grouped by the number of servers $\left|\mathcal{K}\right|$, for the basic and extended \VMCPs, respectively.
As expected,  the CPU time of \CSES generally increases with the number of servers $\left|\mathcal{K}\right|$ for both basic and extended \VMCPs.
This is reasonable as the problem size and the search space grow with the number of servers. 
Nevertheless, even for the largest case ($|\KSET|=1000$), 
\CSES can still solve 95\% of basic \VMCP instances and 87\% of extended \VMCP instances to optimality, respectively, which shows the scalability of the proposed \CSES with the increasing number of servers.
\begin{figure}[!htbp]
	\centering
	\includegraphics[width=0.8\linewidth]{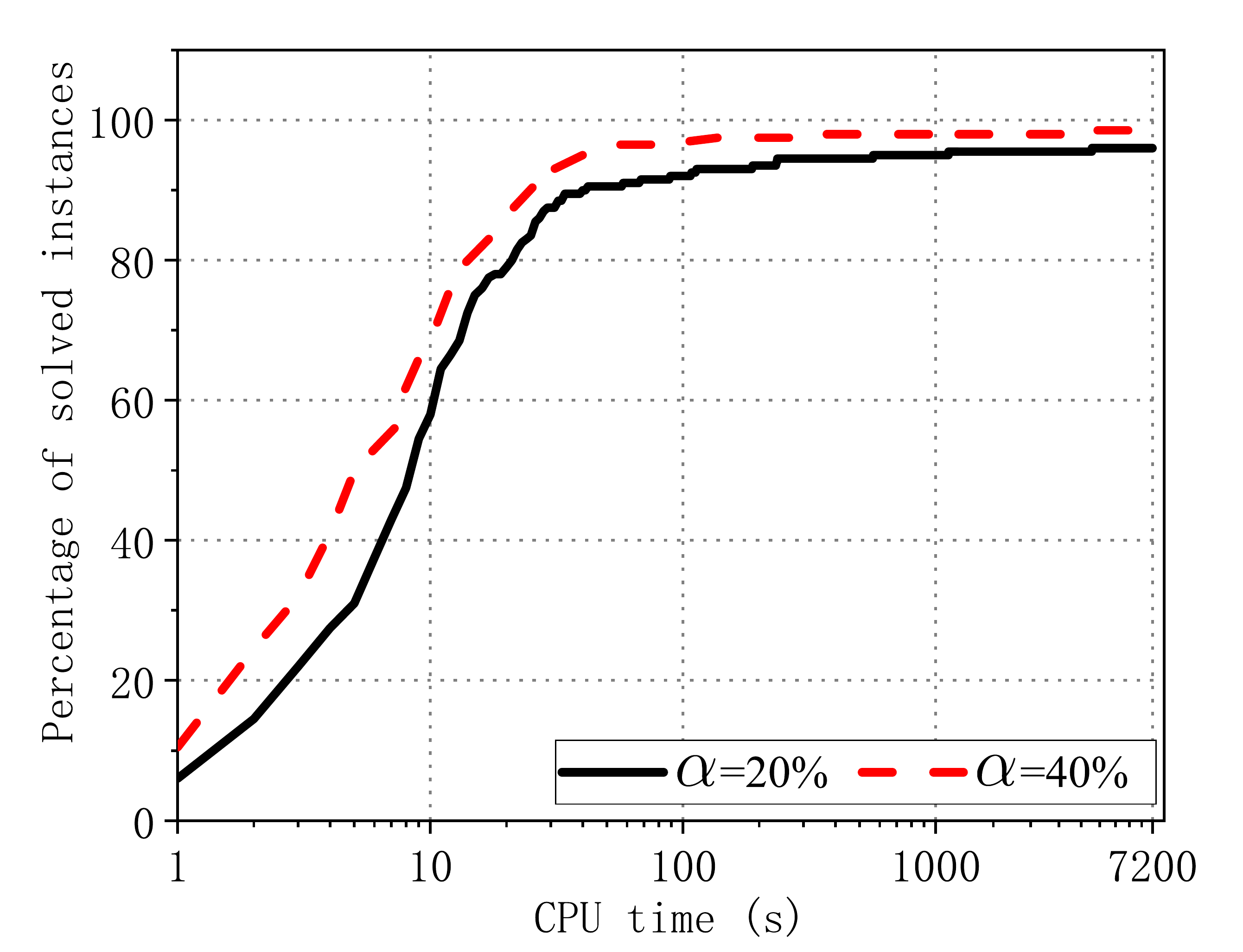}
	\caption{Comparison of the CPU time for \CSES on basic \VMCPs with different loads $\alpha$.} 
	\label{fig:time-beta-1-1}
\end{figure}

Next, we compare the results for basic and extended \VMCPs with loads $\alpha=20\%$ and $\alpha=40\%$.
The results for basic and extended \VMCPs are summarized in Figs. \ref{fig:time-beta-1-1} and \ref{fig:time-beta-2-1}, respectively.
We observe that the CPU time does not increase with the increasing value of $\alpha$ for basic \VMCPs.
Even in extended \VMCPs, the CPU time of solving instances with load $\alpha=40\%$ is only slightly larger than that of solving instances with load $\alpha=20\%$.
However, the same behavior cannot be observed in the computational results returned by \CPX (as illustrated in Table \ref{table:comparison_results_three_model_5}, the CPU time of using \CPX to solve the basic \VMCPs with $\alpha=20\%$ is smaller than that of using \CPX to solve the basic \VMCPs with $\alpha=40\%$).
This shows another advantage of the proposed \CSES, i.e., a higher load $\alpha$ does not lead to a larger CPU time for solving \VMCPs.
\begin{figure}[!htbp]
	\centering
	\includegraphics[width=0.8\linewidth]{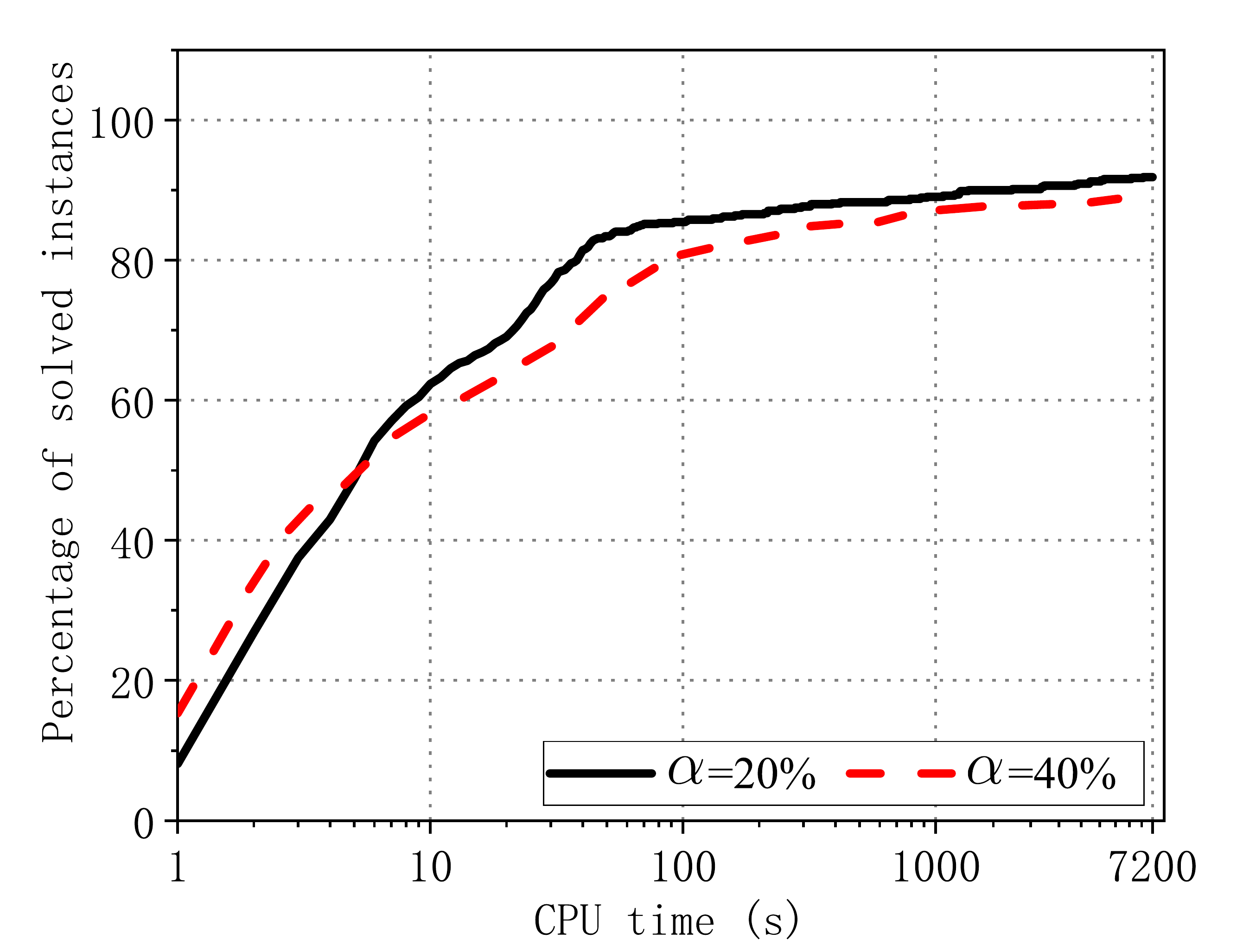}
	\caption{Comparison of the CPU time for \CSES on extended \VMCPs with different loads $\alpha$.} 
	\label{fig:time-beta-2-1}
\end{figure}
\section{Conclusion and remarks}
\label{sect:conclusion_remarks}

In this paper, we have proposed new problem formulations for \VMCP which minimized the summation of server activation, \VM allocation, and migration costs subject to the resource constraints of the servers and other practical constraints.
Compared with existing formulations in 
Speitkamp and Bichler \cite{speitkamp2010mathematical} and Mazumdar and Pranzo \cite{mazumdar2017power} that suffer from large problem sizes due to the 3-index variables, 
the proposed formulation uses the 2-index variables, making a much smaller problem size.
We have developed a cut-and-solve algorithm to solve the new formulations of \VMCPs to optimality.
The proposed algorithm is based on a newly proposed relaxation, which compared with the natural \LP relaxation, is much more compact in terms of providing a better relaxation bound, making it suitable to solve large-scale \VMCPs.
Extensive computational results demonstrate that (i) the proposed formulation significantly outperforms existing formulations in terms of solution efficiency; and (ii) compared with standard \MIP solvers, the proposed \CS algorithm is much more efficient.

\label{sect:appendix}
\begin{proof}
[Proof of Proposition \ref{th:non_trivial_facets}]
	As $a_i > 0$ and $ b> 0$, $\mathcal{X}$ is an independent system and thus every facet-defining inequality of $\conv(\mathcal{X})$, except $x_i \geq 0$, 
	is of the form $\sum_{i \in \mathcal{N}}\pi_{i}x_{i} \leq \pi_{0}$ with $\pi_i \geq 0$ ($i \in \mathcal{N}$) and $\pi_0 > 0$; see \cite[Page 237]{nemhauser1988integer}.
	In addition, $\mathcal{X}_y$ can be transformed to
	$\mathcal{X}_{y^{\prime}}=\left\{(x,y^{\prime})\in \mathbb{Z}_{+}^{\left|\mathcal{N}\right|} \times \left\{0,1\right\}:\sum_{i \in \mathcal{N}}a_{i}x_{i}+by^{\prime} \leq b, x_{i}\leq v_{i},~\forall~i \in \mathcal{N}\right\}$ 
	by replacing variable $y$ with  $1-y' \in \{0,1\}$. 
	Similarly, $\mathcal{X}_{y^{\prime}}$ is an independent system, and thus every facet-defining inequality, except $x_i \geq 0$ and $y'\geq 0$, is of the form $\sum_{i \in \mathcal{N}}\pi_{i}x_{i} + \pi_{0} y' \leq \alpha$ with $\pi_i \geq 0$ ($i \in \mathcal{N}$), $\pi_0 \geq 0$, and $\alpha > 0$.
	This implies that all facet-defining inequalities for $\text{conv}(\mathcal{X}_{y})$ except $x_{i} \geq 0$ and $y \leq 1$ are of the form 
	$\sum_{i \in \mathcal{N}}\pi_{i}x_{i}\leq \pi_{0}y + \alpha - \pi_0$ 
	where $\pi_{i} \geq 0$ ($i \in \mathcal{N}$), $\pi_{0} \geq 0$, and $\alpha > 0$.
	Since $(\boldsymbol{0}, 0) \in \mathcal{X}_y$, we have $\alpha - \pi_0\geq 0$.
	If $\alpha - \pi_0 >0$, inequality $\sum_{i \in \mathcal{N}}\pi_{i}x_{i}\leq \pi_{0}y +\alpha-\pi_0$ can be strengthened to $\sum_{i \in \mathcal{N}}\pi_{i}x_{i}\leq  \alpha y$ and thus cannot be facet-defining for $\conv(X_y)$.
	Consequently, we must have $\pi_0=\alpha > 0$. 
	Next, we shall complete the proof by showing that $\sum_{i \in \mathcal{N}}\pi_{i}x_{i}\leq \pi_{0}y$ with $\pi_i \geq 0$ ($i \in \mathcal{N}$) and $\pi_0 > 0$ (differing from $y \geq 0$) is facet-defining for $\conv(\mathcal{X}_y)$ if and only if $\sum_{i \in \mathcal{N}}\pi_{i}x_{i}\leq \pi_{0}$ with $\pi_i \geq 0$ ($i \in \mathcal{N}$) and $\pi_0 > 0$ is facet-defining for $\conv(\mathcal{X})$.
	
	Suppose that $\sum_{i \in \mathcal{N}}\pi_{i}x_{i} \leq \pi_{0}y$ with $\pi_i \geq 0$ ($i \in \mathcal{N}$) and $\pi_0 > 0$ (differing from $y \geq 0$) is facet-defining for $\conv(\mathcal{X}_y)$.
	Then (i) $\sum_{i \in \mathcal{N}}\pi_{i}x_{i} \leq \pi_{0}$ is valid for $\mathcal{X}$ (as $(x,1) \in \mathcal{X}_y$ if and only if $x \in \mathcal{X}$); 
	and (ii) there must exist $\left|\mathcal{N}\right|+1$ affinely independent points  $(x^\ell,y^\ell ), \ell =1,\cdots,\left|\mathcal{N}\right|+1$, in $\mathcal{F}_y=\left\{(x,y) \in \conv(\mathcal{X}_{y}): \sum_{i \in \mathcal{N}}\pi_{i}x_{i} = \pi_{0}y\right\}$.
	As $(\boldsymbol{0},0)$ is the only point in $\mathcal{F}$ satisfying $y=0$ and $\sum_{i \in \mathcal{N}}\pi_{i}x_{i} \leq \pi_{0}y$ differs from $y\geq 0$, such $|\mathcal{N}|+1$ points must be $({0},0)$ and $({x}^{\ell},1)$, $\ell = 1,\cdots,\left|\mathcal{N}\right|$.
	Apparently,  points ${x}^{\ell}$, $\ell = 1,\cdots,\left|\mathcal{N}\right|$, must be affinely independent and in $\mathcal{F}=\left\{x \in \conv(\mathcal{X}): \sum_{i \in \mathcal{N}}\pi_{i}x_{i} = \pi_{0}\right\}$.
	Therefore, $\sum_{i \in \mathcal{N}}\pi_{i}x_{i} \leq \pi_{0} $ is facet-defining for $\conv(\mathcal{X})$.
	
	Now suppose that $\sum_{i \in \mathcal{N}}\pi_{i}x_{i} \leq \pi_{0}$ with $\pi_i \geq 0$ ($i \in \mathcal{N}$) and $\pi_0 > 0$ is facet-defining for $\conv(\mathcal{X})$.  
	Then  (i) $\sum_{i \in \mathcal{N}}\pi_{i}x_{i} \leq \pi_{0}y$ is valid for $\mathcal{X}_y$ (as $(x,1) \in \mathcal{X}_y$ if and only if $x \in \mathcal{X}$ and $\sum_{i \in \mathcal{N}}\pi_{i}x_{i} \leq \pi_{0}y$ holds at $(\boldsymbol{0},0)$); and (ii) there must exist $\left|\mathcal{N}\right|$ affinely independent points  $x^\ell, \ell =1,\cdots,\left|\mathcal{N}\right|$ in $\mathcal{F}$.
	Apparently, $(\boldsymbol{0},0)$ and $({x}^{\ell},1)$, $\ell = 1,\cdots,\left|\mathcal{N}\right|$, are in $\mathcal{F}_y$, which shows that
	$\sum_{i \in \mathcal{N}}\pi_{i}x_{i} \leq \pi_{0}y $ is facet-defining for $\conv(\mathcal{X}_y)$.
\end{proof}

\section*{Acknowledgement}
This work was partially supported by the Chinese NSF grants (Nos. 1210011180, 12171052, 11971073, 11871115, 12021001, 11991021, and 12201620), and
Alibaba Group through Alibaba Innovative Research Program.




\bibliographystyle{./elsarticle-num}
\bibliography{./ecrc-template}






\end{sloppypar}
\end{document}